\documentclass[11pt,letterpaper]{article}
\usepackage[margin=1in]{geometry}
\usepackage{setspace}
\usepackage{amsmath,amssymb,amsthm}
\usepackage{hyperref}
\usepackage{cleveref}
\usepackage{thmtools}
\usepackage{mathrsfs}
\usepackage[linesnumbered,ruled,vlined]{algorithm2e}
\usepackage{geometry}
\usepackage{graphicx}
\usepackage{enumerate}
\usepackage[colorinlistoftodos, textsize=scriptsize]{todonotes}

\newcommand{\calA}{\mathcal{A}}

\newcommand{\calB}{\mathcal{B}}
\newcommand{\calD}{\mathcal{D}}

\newcommand{\calP}[1]{\mathcal{P}\left(#1\right)}

\newcommand{\filter}{\mathcal{I}}

\newcommand{\whplty}{L^{whp}}
\newcommand{\explty}{L^{exp}}

\newcommand{\low}{\mathrm{low}_{\beta}}
\newcommand{\high}{\mathrm{high}_{\beta}}
\usepackage{xparse}  

\newcommand{\hcm}[1][1]{\hspace*{#1 cm}}

\newcommand{\istrut}[2][0]{\rule[- #1 mm]{0mm}{#1 mm}\rule{0mm}{#2 mm}}

\newcommand{\ignore}[1]{}

\newcommand{\E}{{\mathbb E\/}}

\newcommand{\floor}[1]{\left\lfloor #1 \right\rfloor}
\newcommand{\bydef}{\stackrel{\operatorname{def}}{=}}
\newcommand{\poly}{\operatorname{poly}}
\newcommand{\polylog}{\operatorname{polylog}}
\newcommand{\success}{\operatorname{succ}}

\newcommand{\ContentionResolution}{\textsf{Contention Resolution}}

\NewDocumentCommand{\Expec}{o m}{
  \mathbb{E}
  \IfValueT{#1}{_{#1}} 
  \left[ #2 \right]
}

\NewDocumentCommand{\Prob}{o m}{
  \operatorname{Pr}
  \IfValueT{#1}{_{#1}} 
  \left[ #2 \right]
}

\NewDocumentCommand{\plow}{o}{
  \IfValueTF{#1}
    {p^{\mathrm{low}(#1)}_{\mathcal D}}
    {p^{\mathrm{low}}_{\mathcal D}}
}

\NewDocumentCommand{\slow}{o}{
  \IfValueTF{#1}
    {s^{\mathrm{low}(#1)}_{\mathcal D}}
    {s^{\mathrm{low}}_{\mathcal D}}
}

\NewDocumentCommand{\lowvar}{m o}{
    #1_{\mathcal D}^{\mathrm{low}
        \IfValueTF{#2}
            {(#2)}
            {}
        }
}

\newcommand{\LocalClock}{\textsf{LocalClock}}
\newcommand{\GlobalClock}{\textsf{GlobalClock}}
\newcommand{\local}{\textsf{loc}}

\newcommand{\Bin}{\texttt{Bin}}
\newcommand{\Code}{\texttt{Code}}

\newtheorem{theorem}{Theorem}[section]
\newtheorem{lemma}[theorem]{Lemma}

\newtheorem{fact}[theorem]{Fact}

\newtheorem{conjecture}[theorem]{Conjecture}
\newtheorem{corollary}[theorem]{Corollary}
\theoremstyle{definition}
\newtheorem{definition}{Definition}[section]
\newtheorem{question}{Question}[section]

\usepackage{thmtools}
\usepackage{thm-restate}
\newtheorem*{remark}{Remark}

\title{Contention Resolution, 
With and Without a Global Clock\thanks{This work was conducted while the first three authors were visiting University of Michigan.  Supported by NSF Grants CCF-2221980 and CCF-2446604.}}

\author{\begin{tabular}{c@{\hcm[.8]}c@{\hcm[.8]}c}
  Zixi Cai & Kuowen Chen & Shengquan Du\\
 {\small IIIS, Tsinghua University}
 &
 {\small IIIS, Tsinghua University}
 &
 {\small IIIS, Tsinghua University}\\
 &\\
 Tsvi Kopelowitz & Seth Pettie & Ben Plosk\\
 {\small Bar-Ilan University}
 &
 {\small University of Michigan}
 &
 {\small Bar-Ilan University}
 \end{tabular}}

\date{}

\begin{document}

\maketitle
\thispagestyle{empty}
\begin{abstract}
In the \ContentionResolution{} problem $n$ parties each wish to have exclusive use of a shared resource for one unit of time.  A canonical example is $n$ devices that each must broadcast a packet of information on a shared channel, but the same principles apply to other distributed systems.
The problem has been studied since the early 1970s, under a variety of assumptions on feedback (collision detection, etc.) given to the parties, how the parties wake up (synchronized, adversarial, random), knowledge of $n$, and so on.  The most consistent assumption is that parties do not have access to a \emph{global clock}, only their \emph{local time} since wake-up.  
This is surprising because 
the assumption of a global clock is both 
technologically realistic and algorithmically interesting.
It enriches the problem, and opens the door to entirely new techniques.

\medskip

In this paper we explore the power of the \emph{\GlobalClock} model and establish several new complexity separations, both between \GlobalClock{} and the usual \LocalClock{} model, and within the \LocalClock{} model.  
Our primary results are:

\begin{description}
    \item[\GlobalClock{} vs.~\LocalClock.] We design a new \ContentionResolution{} protocol that 
    guarantees latency $O\left(\left(n\log\log n\log^{(3)} n\log^{(4)} n\cdots \log^{(\log^* n)} n\right)\cdot 2^{\log^* n}\right) 
    \le n(\log\log n)^{1+o(1)}$ in expectation and with high probability.  This already establishes at least a roughly-$\log n$ complexity gap between randomized protocols in \GlobalClock{} and \LocalClock.

    \item[In-Expectation vs.~With-High-Probability.] Prior analyses of randomized \ContentionResolution{} protocols in \LocalClock{} guaranteed a certain latency \emph{with high probability}, i.e., with probability $1-1/\poly(n)$.  We observe that it is just as natural to measure \emph{expected latency}, and prove a $\log n$-factor complexity gap between the two objectives for memoryless protocols.
    The In-Expectation complexity is $\Theta(n \log n/\log\log n)$ whereas the With-High-Probability latency is $\Theta(n\log^2 n/\log\log n)$.  Three of these four upper and lower bounds are new.

    \item[No Universally Optimal Protocols.] Given the complexity separation above, one would naturally want a \ContentionResolution{} protocol that is optimal under \emph{both} the In-Expectation and With-High-Probability metrics.  This is impossible!  It is even impossible to achieve In-Expectation latency $o(n\log^2 n/(\log\log n)^2)$ and With-High-Probability latency $n\log^{O(1)} n$ simultaneously. 
\end{description}
\end{abstract}

\newpage

\section{Introduction}\label{sect:intro}

In the abstract \emph{\ContentionResolution} problem there are $n$ parties, where $n$ is typically unknown, which are interested in monopolizing some shared resource for one unit of time.  In each time step, the only actions the parties can take are to \emph{idle} or \emph{grab} the shared resource.  If exactly one party attempts to grab the shared resource, it succeeds in monopolizing it for that unit of time, 
but if two or more try to grab it, they all fail.  The premier application of abstract \ContentionResolution{} is to facilitate $n$ devices to access a shared communications channel, who each want to transmit a packet of information. 
This application may allow for explicit coordination between the parties, e.g., if they are allowed to transmit extra information (beyond their packets), which can be received by other parties.  
The umbrella of \ContentionResolution{} captures dozens of distinct algorithm design problems, depending on the modeling assumptions 
and metrics of efficiency.  
Let us highlight the key choices that must be made in fixing a model.

\begin{description}
\item[Wake-Up times.] The parties may be woken up synchronously~\cite{BenderFG06}, 
or at times chosen by an adversary, 
or via a statistical process, typically a Poisson point process~\cite{Gallager78,Capetanakis79,TsybakovM81,MoselyH85}.
In this paper we assume adversarial wake-up times.

\item[Feedback.] It is always assumed that every party that attempts to grab the shared resource immediately perceives whether it is successful or not.  
Protocols for which this is the \emph{only} feedback are called \emph{acknowledgment based}.  
In general, the system may provide 
ternary $\{0,1,2+\}$-feedback (aka \emph{collision detection})~\cite{Gallager78,Capetanakis79,TsybakovM81,MoselyH85,GreenbergFL87,BenderKPY18,ChangJP19} 
to all parties indicating that zero, one, or at least two parties attempted to grab the resource, 
or binary $\{0/2+, 1\}$-feedback~\cite{AntaMM13,BenderKKP20,MarcoKS22} indicating 
only failure or success.\footnote{In the context of a shared communication channel, 
some protocols assume that $O(\log n)$ bits can be transmitted in the case of success (i.e., ``1''); see~\cite{BenderKPY18}.}
We work exclusively with acknowledgment-based protocols.

\item[Deterministic vs.~Randomized.] When parties have unique IDs in the range $[N]$, $n\leq N$, 
it is possible to solve \ContentionResolution{} deterministically, 
in time that depends on $n$ and $N$.
See \cite{KomlosG85,Khasin89,ClementiMS01,ChlebusGKR05,MarcoK15} and the references therein.
In contrast, instantiations of randomized algorithms 
are identical, and break symmetry via locally flipping coins. 
Following most work in the area, this paper assumes identical randomized parties.

\item[Finite vs.~Unbounded.] In the $\{0,1,2+\}$- and $\{0/2+,1\}$-feedback models, there exist protocols~\cite{BenderFGY19,ChangJP19,BenderKKP20} that run infinitely, 
without deadlock.  
In this setting \emph{throughput} is the main measure of efficiency.\footnote{I.e., the 
maximum long-term rate of new party wake-ups that can be sustained in such a way that the \emph{backlog} is bounded,
infinitely often.}  Acknowledgment-based protocols are typically analyzed under the assumption that 
$n$ is finite, and usually unknown, where the main metric is \emph{latency} as a function of $n$. 
This is the gap between a party's wake-up time and the time it monopolizes the shared resource.  Our results only consider the finite setting.

\item[Adversarial Power.]
If the wake-up times of the parties are chosen by an adversary, one can still distinguish 
between an \emph{oblivious adversary}, which chooses $n$ and all wake-up times at the beginning,
and an \emph{adaptive adversary}, which decides how many parties to wake up in each time step
based on the state and history of all parties.
Our results hold under the weakest assumption: an adaptive adversary when proving upper bounds, and an oblivious adversary when proving lower bounds.

\item[Global vs. Local Clocks.] Let $t$ be the current time and $t_u$ be the time when $u$ was woken up.  The vast majority of prior work assumes what we call the \emph{\LocalClock} model, in which $u$ perceives only $t-t_u$.  In the \GlobalClock{} model~\cite{GoldbergMPS00,MarcoK15} $u$ perceives both $t$ and $t-t_u$, and can grab the resource as a function of $(t,t-t_u)$, or with a probability that depends on $(t,t-t_u)$.
\end{description}

This list captures only a subset of models from the literature, and leaves out recent work on \emph{resilient} protocols~\cite{BenderFGY19,ChangJP19,BenderFGKY24},
that work even against an adversary that can ``jam'' the shared channel,
as well as \emph{energy efficient} protocols~\cite{BenderKPY18,BenderFGKY24}, that are awake
(receive feedback $\{0,1,2+\}$ or $\{0/2+,1\}$) for a small number of time slots.

\medskip

\paragraph{Maximally Efficient Protocols.} Perhaps the most important take-away message from prior work is that protocols achieving \emph{constant throughput} in the unbounded setting or \emph{linear latency} in the finite setting are only known to be possible in a few situations:

\begin{itemize}
    \item Constant throughput/linear latency is possible in the \LocalClock{} model under a Poisson or adversarial wake-up schedule, as long as $\{0,1,2+\}$- or $\{0/2+,1\}$-feedback is given~\cite{Capetanakis79,MoselyH85,TsybakovM78,TsybakovM81,BenderFGY19,BenderKPY18,BenderKKP20,ChangJP19,BenderFGKY24,MarcoKS22,GreenbergFL87}.

    \item Constant throughput is possible in the acknowledgment-based \GlobalClock{} model if parties are woken up according to a Poisson point process with rate less than $1/e$~\cite{GoldbergMPS00}.
    Goldberg et al.'s~\cite{GoldbergMPS00} protocol makes no latency guarantees.
    
    \item If the $n$ parties ($n$ unknown) wake up synchronously,
    $O(n)$ latency is possible w.h.p.~via the \emph{sawtooth} protocol of Bender et al.~\cite{BenderFHKL04}, which is acknowledgment-based.
    \item If the $n$ parties agree on an approximation $N=\Theta(n)$, 
    then $O(n)$ latency is possible in the \LocalClock{} model via the acknowledgment-based 
    protocol of De Marco and Stachowiak~\cite{MarcoS17}, even under 
    adversarial wake-up times.
\end{itemize}

Because acknowledgment-based protocols are
limited in their ability to perceive contention, 
they seem to be incapable of achieving constant throughput (unbounded) or linear latency (finite) 
under an \emph{adversarial} wake-up schedule.
De Marco and Stachowiak~\cite{MarcoS17} 
proved this in the finite setting and \LocalClock{} model, specifically that against 
an oblivious adversary, the maximum latency of 
some party is $\Omega(n\log n/\log^2\log n)$ with high probability.
In the unbounded setting and \LocalClock{} model, 
Goldberg and Lapinskas~\cite{GoldbergL25} proved that outside of a tiny class of ``LCED'' protocols,\footnote{LCED: \emph{largely constant with exponential decay}} all memoryless\footnote{AKA \emph{backoff}-type protocols, meaning the behavior of a party depends on its local clock, 
not its past actions.} 
acknowledgment-based protocols 
are \emph{unstable} under Poisson arrivals, 
for any rate $\lambda>0$. I.e., with probability 1 the throughput eventually 
goes to 0 and the backlog of unsuccessful parties goes to $\infty$.
In light of these results, we ask the natural question:

\begin{question}\label{q:big-question}
    What are the \underline{\emph{minimal additional assumptions}} 
    necessary to allow 
    a randomized 
    acknowledgment-based protocol
    to achieve linear latency in the finite setting, or
    constant throughput in the unbounded setting?
\end{question}

It is an open question whether assuming the
\GlobalClock{} model \emph{alone} is sufficient to answer \Cref{q:big-question}, that is, if the wake-up schedule is adversarial, whether there is a randomized \GlobalClock{} protocol that achieves linear latency in the finite setting or constant throughput in the unbounded setting.
We are aware of only two previous protocols that assumed global clocks: Goldberg et al.'s~\cite{GoldbergMPS00} constant throughput 
protocol under (non-adversarial) Poisson wake-ups,
and a deterministic protocol of De Marco and Kowalski~\cite{MarcoK15}, where parties have unique IDs in $[N]$.

\subsection{New Results}

In this paper we work exclusively with randomized, acknowledgment-based protocols,
and consider protocols in both the \LocalClock{} and \GlobalClock{} models.
Recall that \emph{acknowledgment-based} means the parties receive \emph{no feedback} except at the moment of their own success.
 All of our algorithms are \emph{memoryless}, meaning their behavior depends only on the clocks (local or local and global), not on the history of their behavior, and our lower bounds apply to memoryless algorithms.  
 (This is a common assumption in the literature.  Goldberg and Lapinskas call these \emph{backoff-type} algorithms; Bender et al.~\cite{BenderFHKL05} call them \emph{Bernoulli} algorithms.) 
 
\paragraph{The \texorpdfstring{\GlobalClock{}}{GlobalClock} Model.}
We give the first randomized 
\ContentionResolution{} protocol for adversarial wake-ups 
that exploits the \GlobalClock{} model.  
\begin{theorem}\label{thm:global-clock}
    There is an acknowledgment-based \ContentionResolution{} protocol in the \GlobalClock{} model that achieves latency $O(n \zeta(4\log\log n)) = n(\log\log n)^{1+o(1)}$ with high probability, where
    \[
    \zeta(x) = (2x)(2\log x)(2\log^{(2)} x)\cdots(2\log^{(\log^* x)} x).
    \]
\end{theorem}
The peculiar $\zeta$ function is derived from 
Elias's $\omega$-code~\cite{Elias75} 
for the integers.
\Cref{thm:global-clock} and the lower bound of De Marco, Kowalski, and Stachowiak~\cite{demarco2022timeenergyefficientcontention}
establish a complexity separation between 
\GlobalClock{} and \LocalClock.
\begin{corollary}
    \GlobalClock{} is \emph{strictly} more powerful than \LocalClock, in the sense that acknowledgment-based 
    \ContentionResolution{} protocols can have latency $n(\log\log n)^{1+o(1)}$ in the former (\Cref{thm:global-clock}), but must have latency $\Omega(n\log n/(\log\log n)^2)$ in the latter~\cite{MarcoS17,demarco2022timeenergyefficientcontention}, with high probability.
\end{corollary}
We believe that \Cref{thm:global-clock} just scratches the surface in terms of new techniques, and that \emph{linear latency} may be within reach.  \Cref{conj:global-clock} lays out what we think should be possible in this model.

\begin{conjecture}\label{conj:global-clock}
Consider the class of randomized acknowledgment-based protocols in the \GlobalClock{} model.
\begin{enumerate}
\item In the finite setting, there exists a 
        protocol with latency $O(n)$ with high probability against an adaptive adversary.
\item In the unbounded setting, there exists a 
protocol that utilizes the shared resource at some rate $\lambda$,\footnote{Here we have in mind the definition of utilization implicit in the analysis of~\cite{ChangJP19}.  Utilization $\lambda$ is achieved if (1) there is a potential function $\Phi$ such that waking up $n'$ parties increases $\Phi$ 
by at most $n'/\lambda$,  
(2) if $\Phi = \Omega(1)$ is sufficiently large, executing the protocol for another time step decreases $\Phi$ 
by $1+\epsilon$ in expectation, $\epsilon>0$, 
and (3) if $\Phi=O(1)$ is bounded, 
the number of 
still-unsuccessful parties left in the system is also $O(1)$.
This definition does not constrain the adversary, e.g., 
by limiting the wake-up rate to 
$\lambda$ over certain windows 
of time~\cite{BenderFHKL05}.} 
even against an adaptive adversary.
\end{enumerate}
\end{conjecture}

\paragraph{The \texorpdfstring{\LocalClock{}}{LocalClock} Model.}
Prior work in the traditional \LocalClock{} model proved that a certain latency $L$ is achieved \emph{with high probability}~\cite{BenderFHKL05,MarcoS17,demarco2022timeenergyefficientcontention}.
It is just as natural to seek protocols that guarantee 
a certain latency \emph{in expectation}.
We establish asymptotically sharp bounds on the 
latency of memoryless protocols under both metrics, 
which reveals a complexity separation between these two 
objective functions.  
Three of the four bounds of \Cref{thm:LocalClock-exp-whp} are new.

\begin{theorem}[In-Expectation vs.~With-High-Probability]\label{thm:LocalClock-exp-whp}
Consider the class of memoryless, acknowledgment-based \LocalClock{} protocols. 
    \begin{enumerate}
        \item There exists a protocol with latency $O(n\log^2 n/\log\log n)$ 
        \textbf{with high probability}, even aginst an adaptive adversary. See De Marco, Kowalski, 
        and Stachowiak~\cite{MarcoS17,demarco2022timeenergyefficientcontention}
        and \Cref{sect:upper-bounds}.
        \item There exists a protocol with latency $O(n\log n/\log\log n)$ \textbf{in expectation}, even against an adaptive adversary.   See \Cref{sect:upper-bounds}.
        \item No protocol has latency $o(n\log^2 n/\log\log n)$ \textbf{with high probability}, even against an oblivious adversary.\footnote{Is saying $g$ cannot be $o(f)$ the same as saying $g=\Omega(f)$?  It depends on who you ask.  See Appendix~\ref{sect:asymptotic-notation-controvery}.}
 See \Cref{sect:local-clock-lower-bounds}.
        \item No protocol has latency $o(n\log n/\log\log n)$ \textbf{in expectation}, even against an oblivious adversary.  See \Cref{sect:local-clock-lower-bounds}.
    \end{enumerate}
\end{theorem}

Given the complexity separation between 
\textbf{In-Expectation} and \textbf{With-High-Probability} 
metrics, the natural followup question is: can we at least 
run \emph{one} protocol that is optimal under \emph{both} metrics?  In most
models of computation the answer would obviously be \emph{yes}:
just run two optimal protocols in parallel/interleaved.  
There is no such generic interleaving method in the \LocalClock{} model,
and in fact it is \emph{impossible} to simultaneously achieve 
optimality under both metrics.  \Cref{thm: tradeoff between exp-latency and whp-latency without a global clock} shows something even stronger, 
that there is essentially no interesting trade-off possible between 
\textbf{In-Expectation} and \textbf{With-High-Probability} guarantees.

\begin{restatable}{theorem}{TradeoffThm}\label{thm: tradeoff between exp-latency and whp-latency without a global clock}
No acknowledgment-based memoryless \ContentionResolution{} protocol in the \LocalClock{} model can simultaneously 
guarantee both $o(n\log^2 n/\log^2 \log n)$ 
latency in expectation and $n\log^{O(1)} n$ latency 
with high probability.
\end{restatable}

\subsection{Organization}

In \Cref{sect:technical-overview} we give a brief technical overview of the paper, focussing mainly on the lower bounds of 
\Cref{thm:LocalClock-exp-whp}(3,4) and \Cref{thm: tradeoff between exp-latency and whp-latency without a global clock}.
In \Cref{sect:models-metrics} we formally define the \LocalClock{} and \GlobalClock{} models, 
protocols in those models, the latency metrics,
and useful notation used throughout the paper.
In \Cref{sect:global-clock} we give the first randomized \ContentionResolution{} algorithm in the \GlobalClock{} model.  
\Cref{sect:local-clock-lower-bounds} presents  
the $\Omega(n\log n/\log\log n)$ and $\Omega(n\log^2 n/\log\log n)$ lower bounds of \Cref{thm:LocalClock-exp-whp}(3,4) in the \LocalClock{} model, 
and \Cref{sect:upper-bounds} presents matching upper bounds,
including an alternate proof of the $O(n\log^2 n/\log\log n)$ upper bound from~\cite{MarcoS17,demarco2022timeenergyefficientcontention}.
We conclude with some open problems in
\Cref{sect:conclusion}.

\section{Technical Overview}\label{sect:technical-overview}

At any particular time step $t$, the \emph{contention} $\sigma$ 
is the sum, over all active parties $u$, that $u$ attempts to grab the shared resource.  The probability of \emph{some} party being successful is roughly $\sigma e^{-\sigma}$ so for a protocol to be efficient it must create many constant-contention slots.  Conversely, in a lower bound, the adversary typically
tries to sustain $\omega(1)$ contention over a long period of time.

\paragraph{\texorpdfstring{\GlobalClock{}}{GlobalClock} Protocol.} 
The problem with binary exponential backoff or any similar protocol
is that it is easy for the (oblivious) adversary to maintain contention $\omega(1)$ 
over a long period of time.  
In \Cref{sect:global-clock} we exploit the fact that
parties have access to a common global clock $t$ by interpreting
a suffix of the bits in the binary representation of $t$ as an integer $a(t)$, and using $a(t)$ to \emph{synchronize} a multiplicative modification to the default grab-probability of binary exponential backoff.  At any particular time we do not know the best integer $\ell$, 
so we let $(a(t))$ cycle through
\emph{all} integers according to a certain schedule determined by
Elias's $\omega$-code~\cite{Elias75}, in which integer $\ell$ appears periodically
with period at most 
$(2\ell)(2\log\ell)(2\log\log \ell)(2\log\log\log\ell)\cdots (2\log^{(\log^{*} \ell)}\ell)$.

\paragraph{\texorpdfstring{\LocalClock{}}{LocalClock} Upper Bounds.} 
It is sometimes difficult to prove upper bounds against adaptive adversaries because the \emph{strategy space} of the adversary is so large.
The main innovation in 
\Cref{sect:upper-bounds} is to
reduce the analysis of contention resolution protocols 
to what we call 
a \emph{counter game}, whose optimal 
strategy is obvious.  By fixing the optimal strategy, counter games can be analyzed in a straightforward fashion with standard Chernoff bounds.

\paragraph{\texorpdfstring{\LocalClock{}}{LocalClock} Lower Bounds.}
\Cref{sect:local-clock-lower-bounds} is the most technically involved.  We build on the lower bound technique of De Marco, Kowalski, and Stachowiak~\cite{MarcoS17,demarco2022timeenergyefficientcontention}, in which we progressively build up an adversary in \emph{layers}.  
Layers 0 and 1 are from~\cite{MarcoS17,demarco2022timeenergyefficientcontention}. 
For this overview we consider the \emph{latency with high probability} metric.

\centerline{\includegraphics[scale=.47]{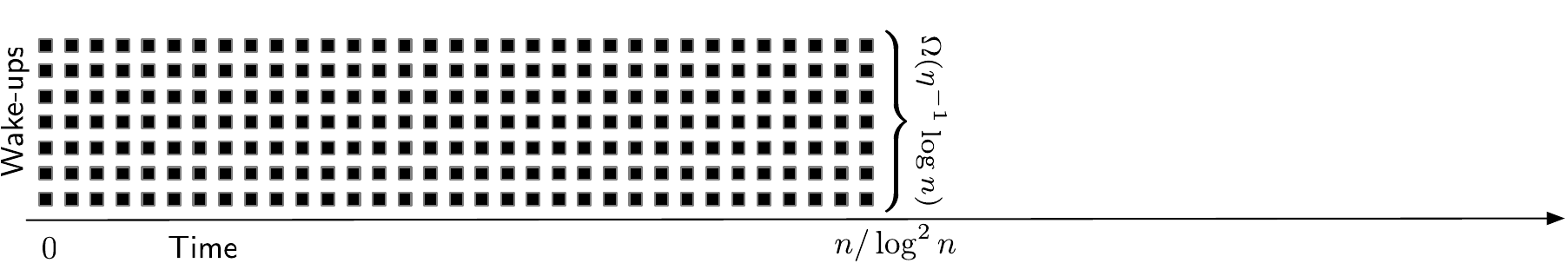}}

\medskip 
\noindent{\emph{Layer 0.}} 
Each party, without loss of generality, attempts to grab the resource at local time 1 with constant probability, say $\eta$.
Thus, by waking up $\Theta(\eta^{-1}\log n)$ parties per time step for $\Theta(n/\log^2 n)$ time steps we
ensure zero successes in this interval, w.h.p., while waking up a negligible fraction of the parties.

\centerline{\includegraphics[scale=.47]{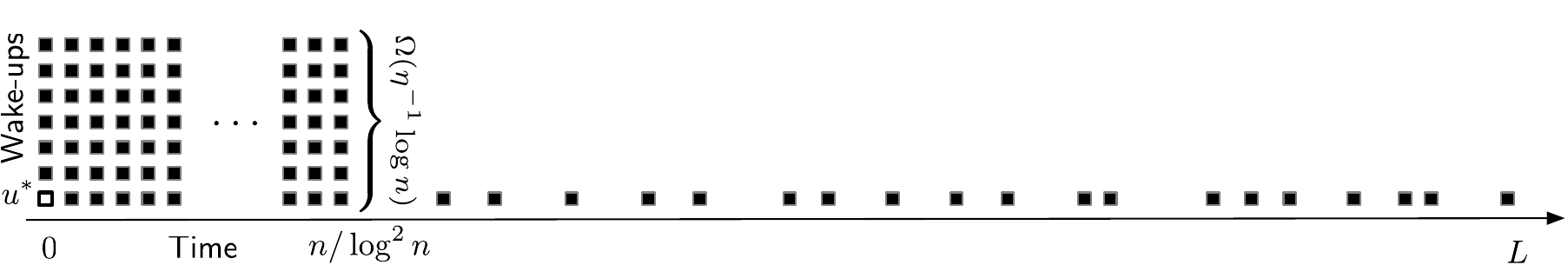}}

\medskip
\noindent{\emph{Layer 1.}}
A most important quantity of a protocol is $s(x)$: the expected number of times a party attempts to grab the resource in its first $x$ times slots.  Identify some party
$u^*$ woken up at time zero, and let
$L=L(n) < n\polylog(n)$ 
be its latency deadline.
We wake up the remaining $n-o(n)$ parties at random times in $[1,L]$.
Every time $u^*$ attempts to 
grab the shared resource in the interval $[n/\log^2 n, L]$, it collides
with a newly awoken party, with probability $1/\polylog(n)$.  
Thus, to achieve success with 
high probability it must be that
$s(L)-s(n/\log^2 n) = \Omega(\log n/\log\log n)$.
We can of course find an $n'$ for which $L(n')=n/\log^2 n$
and conclude that $s(L(n'))-s(n'/\log^2 n')=\Omega(\log n'/\log\log n')$, and so on.
We know that $L(n)<n\polylog(n)$ so by a telescoping sum calculation, 
$s(n)=\Omega((\log n/\log\log n)^2)$.
Thus---ignoring the effect of successes---the average contention over $[0,L]$
is roughly $\Omega(n s(L)/L)$.  
Setting $L=\Theta(n\log n/(\log\log n)^2)$ maintains contention $\Omega(\log n)$ over $[0,L]$, ensuring that there are, in fact, 
no successful parties with high probability.

\medskip 

\noindent{\emph{Layer 2.}} Guaranteeing \emph{no} successes in the interval $[0,L]$ is stronger than necessary when lower bounding the latency $L$.  Suppose we set $L=\Theta(n\log^2 n/ (\log\log n)^3)$, so the average contention
becomes $\Theta(n s(L)/L) = \Omega(\log\log n)$.
However, now each time step sees a success with probability $1/\polylog(n)$, and as the successful party exits the system it 
\emph{reduces} the contention of future time 
steps, thereby \emph{increasing} the probability of success in the future, and so on.  We show that this feedback loop does not spin out of control, and that the contention remains $\Omega(\log\log n)$ over the interval $[0,L]$.
By itself this argument would lead to an $\Omega(n\log^2 n/(\log\log n)^3)$ lower bound.  However, our ambition
is to prove a \emph{sharp} $\Omega(n\log^2 n/\log\log n)$ lower bound, matching the asymptotic complexity of~\cite{MarcoS17,demarco2022timeenergyefficientcontention}.

\centerline{\includegraphics[scale=.47]{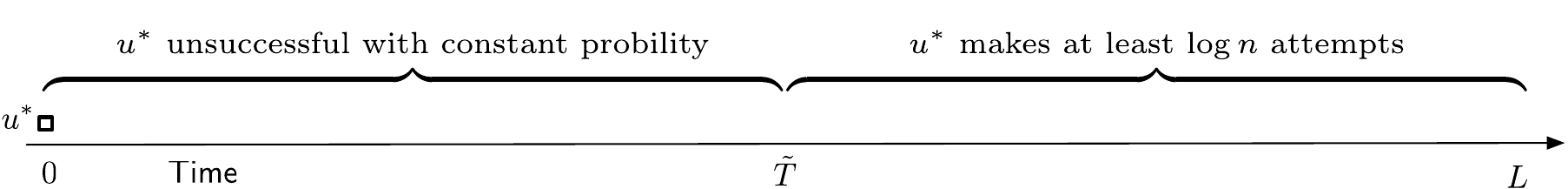}}

\medskip 

\noindent{\emph{Layer 3.}} Define $\tilde{T}$ 
to be such that an adversary that wakes up a constant fraction of the $n$ parties can prevent a party $u^*$ released at time zero from achieving success by time $\tilde{T}$, with constant probability. 
We establish two bounds  capturing the dependence between $s$ and
$\tilde{T}$, which are  informally stated 
below in a \emph{highly simplified} form.
\begin{align}
    \tilde{T} &\geq \frac{ns(n)}{\log\log n},\label{eqn1}\\
    s(L)-s(\tilde{T}) &\geq \log n.\label{eqn2}
\end{align}
These bounds have a \emph{circular} relationship,
in the sense that a weak lower bound on $s(\cdot)$ implies a lower bound on $\tilde{T}$ via \Cref{eqn1}, which then implies an stronger lower bound on $s(\cdot)$ via \Cref{eqn2}.  
This cycling converges on a fixed point of 
$s(n)=\Omega(\log^2 n)$
and $L\ge\tilde{T}=\Omega(n\log^2 n/\log\log n)$.
(The actual statements corresponding to \Cref{eqn1,eqn2} are significantly more complicated than written above, but this discussion captures the flavor of the proof.)

\medskip 

Under different parameterization the same sequence of layers leads to an $\Omega(n\log n/\log\log n)$ lower bound on expected latency, and the impossibility of simultaneous optimality along the In-Expectation and With-High-Probability metrics.

\section{Models, Metrics, and Basic Protocols}\label{sect:models-metrics}

\paragraph{Timing and Clocks.}
Time is partitioned into discrete slots indexed 
by $\mathbb{N}$.  There are $n$ identical parties,
which we index by $[n]$ for notational convenience.
Let $t_u \in [n]$ be the wake-up time of party $u$.
In the \GlobalClock{} model, at time $t$, 
$u$ perceives both $t$ and its local time $t-t_u$, 
whereas in the \LocalClock{} model, it only 
perceives $t-t_u$.

\paragraph{Strategy Space.} 
A \emph{protocol}
in the \LocalClock{} model is a distribution
$\mathcal{D}$ over $\{0,1\}^*$, under
the interpretation that if 
$X^{(u)} = (X_1^{(u)},X_2^{(u)},\ldots) \sim \mathcal{D}$,
then party $u$ will attempt to grab the shared
resource at local time $t_{\local}$ iff $X_{t_{\local}}^{(u)}=1$ \underline{and} $u$ did not achieve
success before local time $t_{\local}$.
A protocol in the \GlobalClock{} model 
is a family of distributions 
$\mathcal{D} = (\mathcal{D}_{t^*})_{t^*\in \mathbb{N}}$ over $\{0,1\}^*$,
where $t^*$ represents the global wake-up time of the party.
In other words, if 
$X^{(u)} = (X_{1}^{(u)},X_2^{(u)},\ldots) \sim \mathcal{D}_{t_u}$, then $u$ attempts
to grab the shared resource at global time 
$t_u + t_{\local}$ iff $X^{(u)}_{t_{\local}}=1$ and $u$ 
has yet to achieve success before this time.

A protocol is called \emph{memoryless} if
$\Pr(X_{t_{\local}}^{(u)}=1)$ depends only on the 
clocks ($t_{\local}=t-t_u$ or both $t,t_{\local}$) 
and not the prior history of party $u$, i.e., $X_1^{(u)},\ldots,X_{t_{\local}-1}^{(u)}$. We often identify a memoryless protocol $\mathcal{D}$ in the \LocalClock{} model 
by the function 
$p(i) \bydef \Pr\left[X_i^{(u)}=1\right]$, where $\mathcal{D}$ is normally omitted.

\paragraph{Active Sets and the Adversary.}
Let $\hat{A}[t]$ be the set of all parties woken up by time $t$, and let $A[t\mid t_0]\subseteq \hat{A}[t]$
be those parties woken up by time $t$ and yet to achieve success by time $t_0$.  
$A[t]$ is short for $A[t\mid t]$.
The number of parties woken up 
at each time step is decided by an 
adversarial strategy $\mathcal{A}$. 
Let $\mathcal{A}[n]$ be the adversary that is constrained to 
wake up $n$ parties.

\paragraph{Latency.}
Let $t_u^{\success}$ be the time when $u$ 
first achieves success, i.e., $X^{(u)}_{t_u^{\success}-t_u}=\sum_{v\in A[t_u^{\success}]} X^{(v)}_{t_u^{\success}-t_v}=1$.
Define $L_{\mathcal{D},\mathcal{A}[n]}^{(u)} = t_u^{\success} - t_u$
to be the \emph{latency} of $u$, 
which is a random variable whose distribution depends on $n$,
the protocol $\mathcal{D}$,
and adversarial 
strategy $\mathcal{A}[n]$.
Minimizing latency is the objective, 
but there are two natural 
ways to look at this metric:
\begin{itemize}
    \item \textbf{Expected Latency.} The goal is to choose 
    $\mathcal{D}$ to minimize the growth of the function
    \[
    L_{\mathcal{D}}^{\exp}(n) = \sup_{\mathcal{A}[n]} \max_{u\in [n]} \E\left[L_{\mathcal{D},\mathcal{A}[n]}^{(u)}\right].
    \]
    In other words, every party should enjoy a bound on its own expected latency, as a function of $n$.
    \item   \textbf{Latency With High Probability.} Given a failure probability threshold $q=q(n)$, the goal is to choose $\calD$ to minimize the latency bound $L_\calD(n,q)$, which satisfies 
    $$
    \forall\calA[n].\; \forall u\in[n].\ \Pr\left[L_{\calD,\calA[n]}^{(u)}\le L_\calD(n,q)\right]\ge 1-q.
    $$

    We define $L_{\calD}^{\mathrm{whp}}(n)=L_{\calD}(n,n^{-2})$, where the error threshold $n^{-2}$ is arbitrary.  Note that $L_{\calD}(n,1/2)\le 2L_{\calD}^{\mathrm{exp}}(n)$ by Markov's inequality.
\end{itemize}

Perhaps the most notorious \ContentionResolution{} protocol is \emph{Binary Exponential Backoff} (BEB). 
Upon waking up at time $t_u$, $u$ partitions all future time slots 
in \emph{windows} $(t_u,t_u+1],(t_u+1,t_u+2],\ldots,(t_u+2^i,t_u+2^{i+1}],\ldots$
and attempts to grab the shared resource at a uniformly random time slot in each window.
Observe that the probability that BEB grabs the resource at time $t_u + t_{\local}$ is $\Theta(1/t_{\local})$.  Thus, BEB is sometimes 
expressed as a \emph{memoryless} protocol with 
$p(t_{\local}) = \Pr\left[X^{u}_{t_{\local}}=1\right] = \Theta(1/t_{\local})$,
independent of $u$'s history.
The simplicity of BEB is attractive, but it is theoretically
undesirable under many metrics.  For example, in the unbounded setting it deadlocks with probability 1 when parties are woken up according
to a Poisson point process with \emph{any} constant rate~\cite{Aldous87,GoldbergL25}.  Even if all parties
are synchronized, its latency is only $O(n\log n)$ 
with high probability, rather than optimal $O(n)$~\cite{BenderFHKL05}.

\medskip

The following notation are 
used in the analysis of 
\LocalClock{} algorithms and 
lower bounds.

\paragraph{Aggregate Contention---Static.} 
Recall that for a memoryless protocol $\mathcal{D}$,
$p(t_\local)=p_{\calD}(t_{\local}) \bydef \Pr\left[X_{t_{\local}}^{(u)}=1\right]$
is the probability of any party $u$ grabbing the resource at local time $t_{\local}$. For a global time slot $t$ and a set of parties $S$, define

$$
\hat\sigma[t;S]=\hat\sigma_{\calA,\calD}[t;S] \bydef \sum_{u\in S}p(t-t_u)
$$
to be the \emph{static aggregate contention} contributed by $S$ at time $t$, with respect to a fixed protocol $\calD$ and 
oblivious adversary $\mathcal A$.  For brevity, we write $\hat\sigma[t] = \hat\sigma[t; \hat A[t]]$. 
We emphasize that $S$ is a \emph{fixed} set chosen before the execution,
making $\hat{\sigma}[t; S]$ static in nature: it is not dependent on the randomness of the parties.

\paragraph{History Prior to Global Time \texorpdfstring{$t$}{t}}
We denote by $H_t$ the \emph{history prior to global time~$t$}, 
that is, the collection of all observations made by all parties before time~$t$. 
In particular, it includes the information 
$(X^{(u)}_{\tau})_{1\le \tau< t-t_u}$ for all $u\in \hat{A}[t]$. We say that a set of parties $S$ is \emph{$H_t$-measurable} if it is a function of $H_t$.



\paragraph{Aggregate Contention---Dynamic.}
When $S=S(H_t)$ is $H_t$-measurable, 
we use different notation to measure \emph{dynamic aggregate contention}.
Define
\[
 \sigma[t; S(H_t)] \bydef \sum_{u\in S(H_t)} p(t-t_u).
\]
For brevity $\sigma[t] = \sigma[t; A[t]]$.  Sometimes we may write
$\sigma[t; S(H_{t_0})]$ to emphasize that $S$ only depends on the history up until time $t_0\leq t$.  It need not be the case that $S\subset A[t]$; it is also useful to measure the aggregate contention that \emph{would} have been contributed by parties $S\subset \hat{A}[t]\backslash A[t]$
who already achieved success before time $t$.

\paragraph{Sum of Access Probabilities.} For a memoryless protocol $\mathcal{D}$, 
define the function $s = s_{\calD}$ as
\begin{align*}
s(k)
&=
\sum_{t_{\local}=1}^k p(t_{\local})
=
\E\left[\sum_{t_{\local}=1}^k X_{t_{\local}}^{(u)}\right], 
& \mbox{which does not depend on $u$.}
\end{align*}
That is, $s(k)$ is the expected number of attempts at the shared resource in the first $k$ steps,
where the subscript may be dropped if implicit.
    
The following lemma is standard in analyzing 
\ContentionResolution{} protocols. The probability of success is maximized when the sum of 
probabilities of parties grabbing the shared resource is constant.  
See Appendix~\ref{sect:prob-of-success} for proof.

\begin{lemma}\label{lem:prob-of-success} 
For a global time $t$, let 
$p_u = \Pr\left[X_{t-t_u}^{(u)}=1\right]$ be the probability that $u$ grabs the shared resource and 
$\sigma = \sigma[t] = \sum_{u\in A[t]} p_u$.

\begin{enumerate}
    \item If $p_u \in [0,1/2]$ for all $u\in A[t]$, then $\Prob{\left(\sum_{u\in A[t]} X_{t-t_u}^{(u)}\right)=1}\geq \sigma 4^{-\sigma}$, 
    and $u$'s probability of success is at 
    least $p_u4^{-\sigma}$.
    \item $\Prob{\left(\sum_{u\in A[t]} X_{t-t_u}^{(u)}\right)=1}\le \sigma e^{-\sigma+1}$
    \ and \ $\Prob{\left(\sum_{u\in A[t]} X_{t-t_u}^{(u)}\right)\le 1}\le e^{-\sigma}+\sigma e^{-\sigma+1}$.
\end{enumerate}
\end{lemma}

\section{Contention Resolution with a Global Clock}\label{sect:global-clock}

Our memoryless algorithm makes use of \emph{Elias codes}, which we now review.  
For a positive integer, let $\Bin(N)\in\{\texttt{0},\texttt{1}\}^*$ be the $(1 + \floor{\log_2 N})$-bit binary representation of $N$.  
Define the sequence 
$N=N_1, N_2, \ldots, N_k=1$ 
where $N_i = \floor{\log_2 N_{i-1}}$.  
Elias's $\omega$-code~\cite{Elias75} assigns a bit-string $\Code(N)$ 
to each positive integer $N$:
\[
\Code(N) = \Bin(N_{k-1})\, \Bin(N_{k-2})\, \cdots\, \Bin(N_1) \, \texttt{0}.
\]
For example, 
$\Code(1) = \texttt{0}$,
$\Code(2) = \texttt{10 0}$,
$\Code(3) = \texttt{11 0}$,
$\Code(4) = \texttt{10 100 0}$,
with spaces introduced for clarity.

\begin{lemma}\label{lem:Elias-code}
    The set $\{\Code(N)\}_{N\in \mathbb{Z}^+}$ is a prefix-free code and the length of 
    $\Code(N)$ is $1 + (1+\floor{\log N}) + (1+\floor{\log\floor{\log N}}) + \cdots + (2)$.
\end{lemma}

We define an infinite sequence $(a(t))_{t\in \mathbb{N}}$ of positive integers indexed by the global time $t$ as follows.
Let $S = \Bin(t)^{R} \texttt{000}\cdots $ be an infinite bit-string obtained by reversing the binary representation of $t$, 
then padding it with an infinite suffix of \texttt{0}s.
Let $\Code(N) = S[0]S[1]\cdots S[|\Code(N)|-1]$ be the prefix
of $S$ in the code book, and set $a(t) = N$.
We will also require 
an infinite sequence $(a'(t))$ 
over $\mathbb{Z}$.  
Define 
\[
a'(t) = (-1)^{a(t)\operatorname{mod} 2}\floor{a(t)/2}.
\]

Recall from the introduction that 
   \(
    \zeta(x) = (2x)(2\log x)(2\log^{(2)} x)\cdots(2\log^{(\log^* x)} x)
    \).
\begin{lemma}\label{lem:a-sequence}
The sequences $(a(t))_{t\in\mathbb{N}}$ and $(a'(t))_{t\in\mathbb{N}}$
have the following properties.
\begin{enumerate}
    \item Let $k$ be a positive integer and $I\subset \mathbb{N}$
    be any interval of width $2^{|\Code(k)|}\leq \zeta(k)$.
    Then $\{1,2,\ldots,k\}$ is a subset of 
    $\{a(t) \mid t\in I\}$.
    \item Let $k$ be a positive integer and $I\subset \mathbb{N}$ 
            be an interval of width
            $2^{|\Code(2k+1)|} \leq\zeta(2k+1)$.
            Then $\{0, \pm 1,\pm 2,\ldots,\pm (k-1),\pm k\}$
            is a subset of $\{a'(t) \mid t\in I\}$.
\end{enumerate}
\end{lemma}

\begin{proof}
    The first claim follows from the fact that $k$ 
    appears in $(a(t))$ periodically with period $2^{|\Code(k)|}\leq \zeta(k)$,
    and that for positive $k'<k$, 
    $|\Code(k')|\leq |\Code(k)|$.
    The second claim follows from the fact that
    $\{1,\ldots,2k+1\}$ are mapped onto
    $\{0,\pm 1,\pm 2,\ldots,\pm (k-1),\pm k\}$.
\end{proof}

\subsection{The Algorithm}

The memoryless version of BEB grabs the shared
resource with probability $1/(t-t_u)$, i.e., inversely proportional to $u$'s time since wake-up.  Define
$\tau(t) = \sum_{u\in A[t]} 1/(t-t_u)$ to be the ``natural'' 
BEB contention at time $t$.
If all parties knew $\alpha = 1/\tau(t)$, they could
grab with probability $\alpha/(t-t_u)$, leading to a constant success probability, 
by \Cref{lem:prob-of-success}.  Using the global clock, 
Algorithm~\ref{alg:2} uses $a'(t)$
to synchronize guesses to $\log \alpha$.
The code is written from the perspective of an arbitrary party $u$.  Until $u$ is successful, the variable $t_u^{\success}$ should be treated as $\infty$.

\begin{algorithm}[H]
    \SetAlgoLined\DontPrintSemicolon
    \caption{\textsf{ContRes}$(t,t_u)$ : $t$ is global time, $t_u$ is $u$'s wake-up time, 
    $t_u^{\success}$ is success time.}\label{alg:2}
     \If{$t \in [t_u,t_u^{\success})$ \tcp*{$u$ is awake, but not yet successful}}
        {$u$ grabs the resource with probability 
        $\displaystyle\min\left\{\frac{1}{2},\frac{2^{a'(t)}}{t-t_u}\right\}$.\;}
\end{algorithm}

Observe that Algorithm~\ref{alg:2} can only be executed in the \GlobalClock{} model as it uses $t$ to synchronize the behavior of the parties.

\begin{lemma}\label{lem:Elias-prob-success}
    Let $\tau(t) = \sum_{u\in A[t]} \frac{1}{t-t_u}$ 
    and $k=\lceil | \log\tau(t)| \rceil$.  Absent new wake-ups, Algorithm~\ref{alg:2} has a constant probability
    of seeing at least one 
    success in the interval $[t,t+\zeta(2k+1)]$.
\end{lemma}

\begin{proof}
    By \Cref{lem:a-sequence}, within the interval $[t,t+2^{|\Code(2k+1)|}) \subseteq [t,t + \zeta(2k+1)]$ there are times $t_1,t_2$
    such that $a'(t_1)=k,a'(t_2)=-k$. (Notice that $\log \tau(t)$ can be either positive or negative.)  Absent
    any new wake-ups, and assuming no successes took place in any time slot $t'\notin \{t_1,t_2\}$, \Cref{lem:prob-of-success} implies
    that Algorithm~\ref{alg:2} has a constant probability
    of success in time slot $t_1$ or $t_2$.
\end{proof}

\begin{corollary}\label{cor:Elias-polylog-contention}
    If $\tau(t) \in [\frac 1 {8c} \log^{-2} n, 8c\log^2 n]$
    then in the time interval 
\[
[t,t+\zeta\left(2 \lceil 2\log\log n +\log 8c \rceil+1\right)]=[t,t+(\log\log n)^{1+o(1)}),
\]
absent new wake-ups, Algorithm~\ref{alg:2} has a constant probability of seeing at least one success.
\end{corollary}

Let $\lambda=\lceil 2\log\log n+\log 8c \rceil$ for some constant $c\geq 1$ controlling the probability of error. 
In the analysis, we partition time into \emph{blocks} with width $ \zeta(2\lambda+1) = (\log \log n)^{1+o(1)}$ .

\Cref{cor:Elias-polylog-contention} says that whenever
$\tau(t)\in [\frac{1}{8c\log^2 n} ,8c\log^2 n]$, we have a constant probability of success in the next block, absent new wake-ups.
Define a block to be \emph{heavy} or \emph{light}, respectively, if the block is free of new wake-ups and 
$\tau(t) > 8c\log^2 n$, or $\tau(t) < \frac{1}{8c\log^2 n}$.
A block that is neither heavy nor light and free of new wake-ups is called a \emph{normal} block.

\begin{lemma}\label{lem:bound-heavy-short-blocks}
    Consider an $n$-party execution of Algorithm~\ref{alg:2} during an interval $[t_0,t_0+W]$ of width $W$.
    Then the number of heavy blocks during this interval
    is $O(\frac{n\log W}{\log^{2} n})$.
\end{lemma}

\begin{proof}
    Consider the sum of $\tau(t)$ over the interval.
\begin{align*}
    \sum_{t\in [t_0,t_0+W]} \tau(t) 
    &\le \sum_{u\in [n]}\sum_{t=\max(t_0,t_u+1)}^{t_0+W} \frac{1}{t-t_u}
    \leq \sum_{u\in [n]} O(\log W) = O(n\log W).
\end{align*} 
    Thus, there are at most $O(\frac{n \log W}{\log^2 n})$ time slots $t$ with $\tau(t) > 8c\log^2 n$.  
\end{proof}

\begin{lemma}\label{lem:light-block-success-prob}
    Suppose party $u$ is active at time $t$, where $t$ is in some light block $B$ and satisfies $a'(t)= \lceil 2\log\log n+\log 4c \rceil$.
    Then, the probability that $u$ successfully transmits at time $t$ is at least $\frac{2c\log^2 n}{t-t_u}$.
\end{lemma}

\begin{proof}
    At time $t$, every active party $v$ transmits with probability $p_v= \frac{2^{a'(t)}}{t-t_v}$.
    Specifically $p_u=\frac{2^{a'(t)}}{t-t_u}>\frac{4c\log^2 n}{t-t_u}$. 
    Since $B$ is a light block, $\tau(t)< \frac{1}{8c\log^2 n}$ and so $\hat{\sigma}[t] < \frac{2^{a'(t)}}{8c\log^2 n}<1/2$, which implies that all parties in $A[t]$ transmit with probability at most $1/2$.
By \Cref{lem:prob-of-success}, $u$ succeeds with probability at least $p_u4^{-1/2}\geq \frac{2c\log^2 n}{t-t_u}$.
\end{proof}

\begin{lemma}\label{lem:bound-medium-blocks}
    There are at most $O(n)$ normal blocks with high probability.
\end{lemma}

\begin{proof}
    Consider a normal block $B_i$ beginning at time $t$, $\tau(t)\in [\frac{1}{8c\log^2 n} ,8c\log^2 n]$ so \Cref{cor:Elias-polylog-contention} implies that there exists some constant $p_i= \Theta(1)$ such that the probability of success in the block is $p_i$. 
    Let $p_{\min}= \min \{p_i \}$ be the minimum constant success probability across all normal blocks. 
    By the Chernoff-Hoeffding bound (\Cref{sect:tail-bounds}), the probability that there are at least 
    $\frac{2n}{p_{\min}}$ normal blocks is $\exp(-\Omega(n))$.
\end{proof}
    
\begin{theorem}\label{thm:global-clock-latency}
    Algorithm~\ref{alg:2} has latency at most $L= O(n\zeta(4\log\log n+O(1))) = n(\log \log n)^{1+o(1)}$, with high probability.
\end{theorem}

\begin{proof}
    For any party $u$, we prove that the latency of $u$ is at most $L=\gamma n\zeta(2\lambda+1)=\gamma  n\zeta(2 \lceil 2\log\log n +\log 8c \rceil+1)$ with high probability, where $\gamma$ is a constant determined later.  

    Consider the interval $(t_u,t_u+L]$. Recall that there are $4$ types of blocks. 
    \begin{itemize}
        \item \textbf{Blocks with new wake-ups:} There are at most $n$ blocks that contain new wake-ups.
        \item  \textbf{Heavy blocks:} By applying \Cref{lem:bound-heavy-short-blocks} to the interval $(t_u,t_u+L]$ of length $L$, there are at most $O(\frac{n\log L}{\log^2 n})=o(n)$ heavy blocks.
        \item \textbf{Normal blocks:} By \Cref{lem:bound-medium-blocks} there are at most $O(n)$ normal blocks throughout the execution of the algorithm, with high probability.
        \item  \textbf{Light blocks:}  The remaining blocks are all light blocks. Since $L= \gamma n\zeta(2\lambda+1)$, we have $\frac{L}{\zeta(2\lambda+1)}= \gamma n$.
        Thus, the number of light blocks is at least $\frac{1}{2}\frac{L}{\zeta(2\lambda+1)}$ when $\gamma$ is sufficiently large.
    \end{itemize}
      By \Cref{lem:light-block-success-prob}, in every light block in which $u$ is active, there exists a time slot $t$ in which $u$ succeeds with probability at least $\frac{2c\log^2 n}{t-t_u}\geq \frac{2c\log^2 n}{L}$. 
      Thus, the probability that party $u$ did not succeed during all the light blocks in which $u$ is active is at most 
      \[
      \left(1-\frac{2c \log ^2 n}{L}\right)^{\frac{L}{2\zeta(2\lambda+1)}} \leq e^{-\frac{c\log^2 n}{\zeta(2\lambda+1)}} \leq  e^{-c\log n} \leq \frac{1}{n^c}.
      \]
\end{proof}

\paragraph{A note on non-linear approximations of \texorpdfstring{$n$}{n}.}
De Marco and Stachowiak~\cite{MarcoS17} proved that if the parties share a common approximation $N=\Theta(n)$, 
then linear latency can be achieved, even in the \LocalClock{} model.
\Cref{thm:global-clock-latency} shows that in the \GlobalClock{} model latency $O(n (\log \log n)^{1+o(1)})$ can be achieved, without any \emph{a priori} knowledge of $n$.

We can modify Algorithm~\ref{alg:2} to achieve latency $O(n\log\log N)$, where $N>n$ is an upper bound common to all parties.  In other words, we can achieve latency $O(n\log\log n)$ even with any quasipolynomial approximation $N < 2^{\log^{O(1)} n}$. Rather than use the $(a'(t))$ sequence, we modify Algorithm~\ref{alg:2} to grab with probability $2^k/(t-t_u)$, where $k$ cycles through $\{-2\log\log N,\ldots,2\log\log N\}$
periodically.
The analysis of \Cref{thm:global-clock-latency} still holds, 
substituting $(2\log\log N+1)$ for $\zeta(2\lambda+1)=(\log\log n)^{1+o(1)}$.

\section{Lower Bounds in the \texorpdfstring{\LocalClock{}}{LocalClock} Model}\label{sect:local-clock-lower-bounds}

In this section, we establish the following lower bound for memoryless protocols. 

\begin{theorem}[Formal Statement of \Cref{thm:LocalClock-exp-whp}(3,4)]\label{thm: main lower bound without global clock}
    There is no 
    acknowledgment-based memoryless \ContentionResolution{} 
    protocol $\calD$
    that satisfies either

\begin{description}
    \item[With-High-Probability Guarantee] $L^{\mathrm{whp}}_{\calD}(n) = o(n\log^2 n/\log\log n)$, or
    
    \item[In-Expectation Guarantee] $L^{\mathrm{exp}}_{\calD}(n) = o(n\log n/\log\log n)$.
\end{description}
\end{theorem}

In the course of proving this theorem, we also obtain the following tradeoff showing that even \emph{near}-optimal expected latency and \emph{near}-optimal high-probability latency cannot be achieved at the same time.

\begin{theorem}[Formal statement of \Cref{thm: tradeoff between exp-latency and whp-latency without a global clock}]
   For any acknowledgment-based memoryless \ContentionResolution{} protocol $\calD$, $L_\calD^{\mathrm{exp}}(n)=o(\frac{n\log^2 n}{(\log\log n)^2})$ and $L_\calD^{\mathrm{whp}}(n)=n\log^{O(1)}n$ cannot be achieved simultaneously.
\end{theorem}

Since the proofs in 
this section only involve oblivious adversaries, without loss of generality we may 
assume the protocol 
satisfies $\eta \bydef p(1) > 0$. 
Before proceeding, we review two key definitions:

\begin{itemize}
    \item $s(k)=\sum_{i=1}^k p(i)$ denotes the total attempt probability of a party in its local time $[1,k]$, which also equals the expected number of attempts made in this interval.
    
    \item $L_{\calD}(n,q)$ is the latency bound such that, for each party, its actual latency is at most $L_{\calD}(n,q)$ with probability at least $1 - q$.
    Throughout this section, $q=q(n)$ is either $q(n)=1/2$ or $q(n)=n^{-2}$ depending on whether we are minimizing expected latency or latency with high probability.
\end{itemize}

  

De Marco, Kowalski, and Stachowiak~\cite{demarco2022timeenergyefficientcontention}
proved that there \emph{does not exist} any acknowledgment-based protocol $\mathcal{D}$ for which
\[
L^{\mathrm{whp}}_{\mathcal{D}}(n) = o\!\left(\frac{n\log n}{(\log\log n)^2}\right).
\]
In the following, we informally outline the key steps of their argument in our notation.

\paragraph{Structural overview of the De Marco, Kowalski, and Stachowiak lower bound~\cite{demarco2022timeenergyefficientcontention}.}
\begin{itemize}
    \item For any protocol $\calD$, they constructed an adversary $\calA$ that ensures $\hat\sigma_{\calA,\calD}[t]=\Omega(\log n)$ for all $t\in[1,\Omega(n/\log^2 n)]$. 
    In this case, the probability of success is $n^{-\Omega(1)}$ in $[1,\Omega(n/\log^2 n)]$.

    \item Focus on a specific party $u^*$ that is activated at time $0$. Since there is no success w.h.p in $\lbrack 1,\Omega(n/\log^2 n)\rbrack$, $u^*$ must make a sufficiently large number of attempts in the remaining time slots to guarantee success with latency $L_{\calD}(n,n^{-2})$. 
    It implies $s(L_\calD(n,n^{-2})) - s(\Omega(n/\log^2 n))\ge \Omega(\log n/\log\log n)$. 
    \item Suppose $L_{\calD}(n,n^{-2}) = O(n\polylog(n))$, so if we let $n'<n$ be such that $L(n')=n/\log^2 n$, then $n'=n/\polylog(n)$, 
    and by the same argument above, $s(L(n'))-s(\Omega(n'/\log^2 n'))=\Omega(\log n'/\log\log n')$, and so on.
    By a telescoping-sum argument, $s(n) = \Omega\left(\frac{\log^2 n}{(\log \log n)^2}\right)$. Based on this result, they construct an adversary $\calA'$ such that $\hat\sigma_{\calA',\calD}[t]=\Omega(\log n)$ for all $t\le \Omega(n\cdot s(n)/\log n)=\Omega(n\log n/(\log\log n)^2)$. Hence, with high probability, there is no success in the first $\Omega(n\log n/(\log\log n)^2)$ time slots. 
\end{itemize}

\paragraph{Key intuitive observations.}
For memoryless protocols, we refine the previous lower-bound framework by making the following two observations:

\begin{itemize}
  \item 
  The bottleneck in prior arguments is the requirement that $\hat{\sigma}[t]=\Omega(\log n)$
  to guarantee \emph{no success} with high probability.
  In contrast, our adversary maintains a weaker condition
(approximately $\hat{\sigma}[t]=\Omega(\log\log n)$, though there are additional constraints), 
which permits $o(n)$ of the parties to succeed.  The difficulty here is showing that the (dynamic) aggregate contention $\sigma[t]$ remains $\Omega(\log\log n)$ as well, 
as successful parties effectively remove aggregate contention from the system.
This part of the analysis is enabled by a new tool of \emph{probability thresholds} and \emph{filtering}, which we discuss shortly.
If only these ideas were applied, 
we would end up with a lower bound of $L_{\calD}^{\mathrm{whp}}(n) = \Omega(n\log^2 n/(\log\log n)^3)$.

\item Our lower bound on $s(n)=\Omega(\log^2 n/(\log\log n)^2)$ turns out to be loose by a $(\log\log n)^2$ factor for bounding $L_{\calD}^{\mathrm{whp}}(n)$.  
In order to close this $(\log\log n)$-factor gap 
we need a more refined analysis.  
We define $\tilde{T}_{\calD}(n)$ to be such that, with constant probability, 
we can keep one particular party from achieving success in $\tilde{T}$ time slots 
by waking up a constant fraction of the parties.  We can simultaneously lower bound $\tilde{T}$ in terms of $s$, and lower bound
$s(L_{\calD}^{\mathrm{whp}}(n))-s(\tilde{T})$ by $\Omega(\log n)$ rather than $\Omega(\log n/\log\log n)$.  Analyzing these two lower bounds in tandem ultimately lets us shave off the last two $\log\log n$ factors and achieve a sharp lower bound.
\end{itemize}

\paragraph{Overview and organization.}
The rest of this section is organized as follows:

\begin{itemize}
    \item In Section~\ref{subsec:prob-thresholds}, we introduce a family of thresholds $B_{\beta}(i) = \Theta(\log^{\beta} i / i)$ 
for a fixed constant~$\beta$ and for all $i \ge 1$. 
Most known protocols are smooth and monotonically decreasing. 
For example, the Binary Exponential Backoff (BEB) protocol corresponds to 
a memoryless protocol with $p(i) = \Theta(1/i)$, 
whereas the high-probability protocols of~\cite{MarcoS17,demarco2022timeenergyefficientcontention} 
use $p(i) = \Theta(\log i / i)$. 
In general, a \emph{high-probability slot} (i.e., one with $p(i) > B_{\beta}(i)$) 
tends to have an undesirable effect on the analysis. 
Therefore, we introduce a filter function to exclude time slots 
whose transmission probabilities are excessively large, 
so that we can better capture the essential behavior of the protocol.

    \item Following the above argument, we  explore the properties of probability thresholds. Throughout this section, we mainly focus on the following two functions:
    \begin{itemize}
        \item $s^{\low}(k)=\sum_{1\le i\le k} p^{\low}(i)$, where $p^{\low}(i)=p(i)\cdot \mathbf{1}\{p(i)\le B_{\beta}(i)\}$ for $i\ge 1$. $s^{\low}(k)$ represents the prefix sum of low probabilities. 
        
        \item To estimate the lower bounds, we identify one arbitrary party $u^*$, which is activated at time $0$, and we focus on maximizing the latency of $u^*$. 
        For these reasons, we introduce  $\tilde{T}_{\calD}(n)$ which represents the maximum number of time slots such that an adversary can use at most a constant fraction of the $n$ parties to ensure that $u^*$ cannot succeed by time $\tilde{T}_{\calD}(n)$ with a constant probability. This function is a lower bound for \emph{both} expected and high-probability latency because it only guarantees the latency of one \emph{specific} party.
    \end{itemize}
    \item In \Cref{subsec:construct-adversary}, 
    we construct 
    in \Cref{thm: high-contention construction}
    an adversary that ensures that
    \begin{align*}
    \hat{\sigma}^{\low}[t] =\sum_{u\in \hat{A}[t]} p^{\low}(t-t_u) &\ge \Omega(\log\log n),\\
    \text{for all } t &\le \Theta\left(\frac{n(s^{\low}(\lfloor n/\log^2 n\rfloor)-s^{\low}(\lfloor \sqrt{n}\rfloor))}{\log\log n}\right).
    \end{align*} 
    
    Note that \Cref{thm: high-contention construction} refers to \emph{static} contention $\hat{\sigma}^{\low}[t]$.
    \Cref{thm:restricted-window-lower-bound} extends the lower bound of \Cref{thm: high-contention construction} to show that
    $\tilde{T}_{\calD}(n)$ is also $\Omega\left(\frac{n(s^{\low}(\lfloor n/\log^2 n\rfloor)-s^{\low}(\lfloor \sqrt{n}\rfloor))}{\log\log n}\right)$. 
    The proof is involved, and is postponed to \Cref{subsec:proofofkey}.
    We also prove that whenever the overall latency is $O(n\log^{O(1)} n)$, there are only a small number of high-probability time slots, i.e., those $i$ for 
    which $p(i)>B_\beta(i)$. 
    \item In~\Cref{subsec:analysis-for-lower-bounds}, we show that 
    \[
    s^{\low}(L_{\calD}(n,q(n)))-s^{\low}(\tilde{T}_{\calD}(n))\ge \Omega(\log (1/q(n))).
    \]
    The rationale is that some party $u^*$ inserted
    at time $t_{u^*}=0$ survives to time $\tilde{T}_{\mathcal{D}}(n)$ with constant probability.  In order for it to succeed by 
    its deadline $L_{\mathcal{D}}(n,q(n))$, 
    it must make a sufficient number of attempts
    in the interval $(\tilde{T}_{\mathcal{D}}(n),L_{\mathcal{D}}(n,q(n))]$.  There are
    a sublinear number of high-probability slots that can be blocked by new wake-ups, so the remaining attempts must be counted by $s^{\low}$.
    
    \item Combining the above two results, we observe the following circular relationship:
\begin{itemize}
    \item An improvement in the lower bound of $s^{\low}$ immediately leads to an improvement in the lower bound of $\tilde{T}_{\calD}$.
    \item Conversely, a stronger lower bound on $\tilde{T}_{\calD}$ also yields a stronger lower bound on $s^{\low}$.
\end{itemize}
The process described above eventually converges to a fixed point where both bounds stabilize, and, 
depending on $q(n)\in\{1/2,n^{-2}\}$, implies $\Omega(n\log n/\log\log n)$ and $\Omega(n\log^2 n/\log\log n)$ lower bounds on latency In-Expectation and With-High-Probability,
respectively 
(\Cref{thm: main lower bound without global clock}).
    \item In~\Cref{subsec:tradeoff}, we use a similar argument for the lower bounds: If $L_{\calD}^{\mathrm{whp}}(n)$ is small, then $s^{\low}$ must be large, which in turn implies a lower bound on the expected latency $L_{\calD}^{\mathrm{exp}}(n)$. This reveals the impossibility to achieve near-optimal latency, in expectation and with high probability simultaneously.
    \item Finally, in~\Cref{subsec:proofofkey}, we return to prove~\Cref{thm:restricted-window-lower-bound}, 
    by analyzing the random process induced by the adversary from~\Cref{thm: high-contention construction}. 
    To be more specific, we use the following steps: 

\begin{itemize}
    \item We prove that if $\sigma[t] = \Omega(\log \log n)$ for all $t \in [t_0, t_0 + \delta)$, conditioned on the observed history before $t_0$, then the number of successes in this interval is at most $\frac{\delta}{\log^{\Omega(1)} n}$, where $\delta=\polylog(n)$. 
    We refer to this property as the \emph{low-density property}.
    
    \item We divide the time slots into length $\delta$ segments, and show that if all previous segments satisfy the low-density property, then for any $t$ in the next segment, $\sigma[t] \ge \frac 45\hat\sigma[t]$ when conditioned on the history of previous segments, and thus the next segment also satisfies the low-density property with high probability. This inductive step establishes that all segments are low-density with high probability. As a result, we conclude that $\sigma[t] = \Omega(\log \log n)$ holds with high probability for every $t$.
    
    \item Finally, we show that the expected number of successes of a specific party $u^*$ is at most $\log^{-\Omega(1)} n$, implying that $u^*$ fails with probability $1 - \log^{-\Omega(1)} n$. This completes the derivation of the lower bound on $\tilde T_{\calD}(n)$.
\end{itemize}
\end{itemize}

\subsection{Probability Thresholds}\label{subsec:prob-thresholds}

To construct the adversary, we activate parties at random time slots so that $\E[\hat\sigma[t]]=\Omega(\log\log n)$. We then seek an implementation where $\hat\sigma[t]=\Omega(\log\log n)$ holds for every time slot $t$. However, when $p(t)$ is large, we cannot directly apply Chernoff and union bounds to obtain such an implementation, because the realized value of $\hat\sigma[t]$ can fall significantly below its expectation with probability $e^{-\Theta(\log\log n)}\gg n^{-1}$. 
On the other hand, if each random variable takes values in a small range $[0,v]$, 
we can strengthen the probability bound to $e^{-\Omega(\log\log n / v)}=n^{-\omega(1)}$ when $v$ is sufficiently small. 
To leverage this observation, we introduce \emph{probability thresholds}, 
which allow us to filter out high-probability time slots and thereby obtain tighter failure bounds.


\begin{definition}[Filter Functions] 

 A \emph{filter function} is a function \( \filter: \mathbb{N} \to \{0,1\} \) that maps a \emph{local} time index 
 to an indicator value. Given such a function \( \filter \), a protocol $\calD$, and an oblivious adversary $\calA$, we define the following filtered quantities:
\begin{itemize}
\item $p^{\filter}(i)=p(i)\filter(i)$
\item $s^{\filter}(k)=\sum_{i=1}^k p^{\filter}(i)$
\item $\hat\sigma^{\filter}[t;S]= \E\left[\sum_{u\in S}X_{t-t_u}^{(u)}\filter(t-t_u)\right]=\sum_{u\in S}p^{\filter}(t-t_u).$
\item $\sigma^{\filter}[t;S(H_t)]=\E\left[\sum_{u\in S(H_t)}X_{t-t_u}^{(u)}\filter(t-t_u)\right]=\sum_{u\in S(H_t)}p^{\filter}(t-t_u).$ 
\end{itemize}
\end{definition}

Recall that the memoryless version of BEB is defined by $p(i)=\Theta(1/i)$.  The high-probability protocol of~\cite{MarcoS17,demarco2022timeenergyefficientcontention}
uses $p(i)=\Theta(\log i/i)$.  However a general memoryless
protocol does not necessarily use a $p$ that is simple to define
nor monotone decreasing.  We use a filter function to filter
out slots whose probability is \emph{too large}.

\begin{definition}[Low-probability Threshold]
Given a constant $\beta\ge e$, for $t\ge 1$, we define 
\[
B_\beta(t)=\begin{cases}1&t\le 16\\\min(1,\frac{\ln^{\beta} t}{t})& t>16\end{cases}.
\]
For a acknowledgment-based protocol $\calD$, we further define the following notions:
\begin{itemize}
\item $\mathrm{low}_{\beta}(t)=\mathbf{1}\{p(t)\le B_\beta(t)\}$ is the low-probability filter with respect to $\beta$ and $\calD$.
\item $
\mathrm{low}_\beta^{(\ge m)}(t)=\mathrm{low}_{\beta}(t)\cdot \mathbf 1\{t\ge m\}$ additionally filters out time slots preceding the $m$th.
\item Let $N^{\high}(k)=\sum_{t=1}^{k} (1-\low(t))$ be the number of 
high-probability slots in the 
range $[1,k]$.
\end{itemize}

\end{definition}
$B_{\beta}$ is chosen so that $N^{\mathrm{high}_{\beta}}(n\log^{O(1)}n)$ can be bounded by $o(n/\log n)$, as shown in the following lemma:
\begin{lemma}\label{lem:basic-filter}
    Given $\beta\ge e$, we have the following:
    \begin{enumerate}
    \item $B_{\beta}(t)\le B_{\beta}(t-1)$ for $t\ge 2$. That is, $B_{\beta}(t)$ is monotone non-increasing.         
    \item For any protocol $\calD$, $N^{\high}(k)\le s(k)/B_{\beta}(k)=\max(1,k\ln^{-\beta} k)s(k)$ for $k\ge 1$.
    \end{enumerate}
\end{lemma}
\begin{remark}
    The argument in \cite{MarcoS17,demarco2022timeenergyefficientcontention} shows that $s(n)=\log^{\Theta(1)}n$ when the latency is $n\log^{O(1)}n$. Hence, $N^{\mathrm{high}_\beta}(L(n,q(n)))=o(n/\log n)$ for sufficiently large constant $\beta$.
\end{remark}
\begin{proof}
\mbox{}\\
\noindent
\emph{Proof of Part 1:}
    We compute
    $$
    \frac{\mathrm d}{\mathrm dt}\frac{\ln^\beta t}{t}=\frac{(\beta-\ln t)\ln^{\beta-1}t}{t^2}.
    $$
    This derivative is positive for $t<e^\beta$ and negative for $t>e^\beta$, so $\ln^\beta t/t$ increases on $[1,e^\beta]$ and decreases on $[e^\beta,\infty)$. Thus, $B_\beta(t)$ is non-increasing on $[e^\beta,\infty)$. 
    
    Now consider $t\in[e^e,e^\beta]$. By definition, $B_\beta(t)=1$ if $\ln^{\beta} t\ge t$, which is equivalent to $\frac{\ln t}{\ln\ln t}\le\beta$. Since $\frac{\ln t}{\ln\ln t}$ increases on $[e^e,\infty)$, we have $\frac{\ln t}{\ln\ln t}\le \frac{\ln e^\beta}{\ln\ln e^\beta}=\beta/\ln \beta\le\beta$. Hence $B_\beta(t)=1$ for all $t\in[e^e,e^\beta]$. Therefore, $B_\beta(t)$ is monotonically nonincreasing.

\medskip
\noindent
\emph{Proof of Part 2:}
    For any $t\in[1,k]$, $\mathrm{low}_\beta(t)=0$ implies $p(t)>B_\beta(t)\geq B_\beta(k)$. 
    The number of such positions is thus at most 
    \[
    \frac{s(k)}{B_\beta(k)}=\max\left(1,\frac{k}{\ln^{\beta}k}\right)s(k).
    \]
\end{proof}

To analyze the latency, we introduce the following notion:
\begin{definition}[Restricted Time Window]
\label{def:restricted-window}
Given a protocol $\mathcal{D}$, we define the
\emph{restricted time window} $\tilde{T} = \tilde{T}_{\mathcal{D}}(n)$ as the largest integer $\tilde{T}$
for which there exists an adversary $\mathcal{A}$ such that:

\begin{enumerate}
    \item there exists a party $u^*$ activated at time $0$ that succeeds in the
    interval $[1,\tilde{T}]$ with probability at most $1 - 4^{-1/8} \le 0.16$ under $\mathcal{A}$. 
    \item at most $2n/3$ parties are activated in $[0,\tilde T]$.
\end{enumerate}
\end{definition}

Note that $\tilde{T}_{\calD}(n) < L_{\calD}(n,1/2) < L_{\calD}(n,n^{-2})$
lower bounds both latency bounds of interest.  The following simple adversary ensures no success occurs 
in the first $n/\log^2 n$ slots 
with high probability.

\begin{lemma}\label{lem:simple-adversary}
Fix any acknowledgment-based protocol $\calD$ with $\eta = p(1)>0$
and any sufficiently large integer $n\ge n_0(\eta)$.  Then
\begin{enumerate}
    \item If there exists an adversary that ensures $\hat\sigma[t]\ge 10\ln n$ for $t\in \lbrack 1,T_0\rbrack$ for an integer $T_0\le n^2$, then under this adversary, $\Pr[\exists t\in[1,T_0].\, \sum_{u\in A[t]}X^{(u)}_{t-t_u}\le 1]\le \frac{1}{n^8}$.
    \item There exists an adversary that uses $\lfloor n/3\rfloor$ parties and guarantees that $\hat\sigma[t]\ge 10\ln n$ for $t\in \lbrack 1,\lfloor n/\ln^2 n\rfloor\rbrack$. Thus there are no successes in the interval $\lbrack 1,\lfloor n/\ln^2 n\rfloor$ with high probability.
\end{enumerate}
\end{lemma}
\begin{proof}
\mbox{}\\
\noindent
\emph{Proof of Part 1:}
By~\Cref{lem:prob-of-success}, we know:
\begin{align*}
\Prob{\exists t\in[1,T_0].\: \sum_{u\in A[t]}X^{(u)}_{t-t_u}\le 1} & 
\le \sum_{t\in[1,T_0]} \Prob{\sum_{u\in A[t]} X_{t-t_u}^{(u)}\le 1 \;\middle|\; A[t]=\hat{A}[t]}\\
&\le T_0\cdot (e^{-\sigma[t]}+\sigma[t]e^{-\sigma[t]+1})\\
&\le T_0 \cdot \frac{10e\ln n+1}{n^{10}}
\:\le\: \frac{1}{n^{8}}.
\end{align*}

\emph{Proof of Part 2:}
We activate $\lceil 10\ln n/\eta\rceil$ parties for each time slot $0\le t< \lfloor n/\ln^2 \rfloor$, then we can notice that:
\begin{itemize}
\item The total number of parties is $\lceil 10\ln n/\eta\rceil\cdot \lfloor n/\ln^2 n\rfloor\le n\cdot \frac{10\ln n+\eta}{\eta\ln^2 n}\le n/3$.
\item For each $t\le T_0$, $\hat\sigma[t] \ge \lceil 10\ln n/\eta\rceil\cdot \eta\ge 10\ln n$. Based on \emph{Part 1}, we know there is no successes in $[1,\lfloor n/\ln^2 \rfloor]$ with probability at least $1-n^{-8}$. 
\end{itemize}
\end{proof}
The next lemma provides the first lower bound of $\tilde T_{\calD}(n)$, generalizing~\cite[Lemma 4.6]{demarco2022timeenergyefficientcontention}.
\begin{lemma}\label{lem:restricted-window}
Fix constants $\beta\ge e$, $\eta \in (0,1]$, and error probability $q(n)$, which is either $q(n)=1/2$ or $q(n)=n^{-2}$. 
For any protocol $\calD$ with $\eta \bydef p(1)>0$, 
we have the following:
\begin{enumerate}
\item Trivially, $\tilde{T}_{\calD}(n)< L_{\calD}(n,q(n))$ for $n\ge 1$.
\item $\tilde{T}_{\calD}(n)\ge \frac{s(\lfloor n/\ln^2 n\rfloor) \cdot n}{40\ln n}\ge \frac{n}{\ln^2 n}\,\,\,\text{for sufficiently large }n$.
\item If $L_{\calD}(n,q(n))= o(n\log^{\beta/2-1} n)$, then $N^{\high}(\lceil n\log^{\beta/2-1}n\rceil)\le o(n/\log n)$.
\end{enumerate}
\end{lemma}
\begin{proof}
\mbox{}\\
\noindent
\emph{Proof of Part 1:} By the definition of $\tilde{T}_{\calD}(n)$, there is an adversary that with probability $4^{-1/8}\ge q(n)$, $u^*$ does not succeed by time $\tilde{T}_{\calD}(n)$. Therefore, $\tilde{T}_{\calD}(n)< L_{\calD}(n,q(n))$.  

\medskip
\noindent
\emph{Proof of Part 2:}
Let $T_0=\lfloor n/\ln^2 n\rfloor$, $T_1=\lfloor \frac{s(T_0) \cdot n}{40\ln n}\rfloor$. 
We construct an oblivious adversary $\calA$, which partitions the parties into two sets $S_1$ and $S_2$, each of size $\lfloor n/3\rfloor$.

The wake-up times of the parties in $S_1$ are assigned as in~\Cref{lem:simple-adversary}.
Turning to $S_2$, we assign $\lfloor n/3\rfloor$ parties wake-up times that are independent and uniformly distributed in $[0,T_1)$. 
Then, for each time slot $t\in (T_0,T_1]$,

$$\Expec{\hat\sigma[t]}\geq \sum_{u\in S_2}\Expec{p(t-t_u)}=\sum_{u\in S_2}\frac{1}{T_1}\sum_{t_u=0}^{t-1}p(t-t_u)=\frac{|S_2|}{T_1}s(t)\ge \frac{\lfloor n/3\rfloor}{T_1} s(T_0)\ge 12\ln n.$$

Applying a Chernoff bound (see \Cref{coro:chernoff-bound-2}), $\Pr[\hat\sigma[t]\ge 4\ln n]\ge 1-n^{-2}$. By a union bound, $\Pr[\forall t\in (T_0,T_1].\: \hat\sigma[t]\ge 4\ln n]\ge 1-(T_1-T_0)n^{-2}>0$. 
The upshot is that there exists wake-up schedule such that $\hat\sigma[t]\ge 4\ln n$ for all $t\in (T_0,T_1]$.

Now combine the assignment for $S_1$ and $S_2$. We observe that the probability of seeing \emph{at least one} success in the interval $[1,T_1]$ remains unchanged if successful 
parties do not halt.
That is, even previously successful $u \in \hat{A}[t]$ grab the shared resource whenever 
$X^{(u)}_{t - t_u} = 1$. 
Thus for sufficiently large $n$, 

\begin{align*}
&\Prob{\forall t\in [1,T_1].\; \sum_{u\in A[t]} X^{(u)}_{t-t_u}\not =1}\\&\ge 1-\left(\sum_{t\in (T_0,T_1]}\Prob{\sum_{u\in \hat{A}[t]} X^{(u)}_{t-t_u}=1}\right)-\left(\Prob{\exists t\in[1,T_0].\; \sum_{u\in A[t]} X^{(u)}_{t-t_u}\le 1}\right)\\
&\ge 1-\left(\sum_{t\in (T_0,T_1]} \sigma[t]\cdot e^{-\hat\sigma[t]+1}\right)-\left(\Prob{\exists t\in [1,T_0] .\; \sum_{u\in A[t]} X^{(u)}_{t-t_u}\le 1}\right)\\
\intertext{By~\Cref{lem:simple-adversary} (1) and the fact $\hat\sigma[t]\ge 4\ln n$ for $t\in (T_0,T_1]$,}
&\ge 1-\frac{ns(\lfloor n/\ln^2 n\rfloor)}{40\ln n}\cdot \frac{4e\ln n}{n^4}-\frac{1}{n^8}\\
&\ge 1-\frac{1}{n^2}.
\end{align*}

Therefore, the adversary 
$\calA$ wakes up at most $\frac{2}{3} n$ parties, and with probability at least $1-\frac{1}{n^2}$, there is no success in $[1,T_1]$. This gives us a weak lower bound of $\tilde{T}_{\calD}(n)\ge T_1=\frac{ns(\lfloor n/\ln^2 n\rfloor)}{40\ln n}$ 
for sufficiently large $n\ge n_0(\eta)$.
In addition, since $s(\lfloor n/\ln^2 n\rfloor)\ge \eta$, $\frac{n\cdot s(\lfloor n/\ln^2 n\rfloor}{40\ln n}\ge \frac{n}{\ln^2 n}$ for sufficiently large $n$.

\medskip

\noindent
\emph{Proof of Part 3:} 
Assume $n$ is sufficiently large. We select $n'=\Theta(n\log^{\beta/2+1} n)$ such that $\lceil n\log^{\beta/2-1} n\rceil\le \lfloor n'/\ln^2 n'\rfloor$. By Parts 1 and 2
and the assumption that
$L_{\calD}(n,q(n))=o(n\log^{\beta/2-1} n)$, we know \begin{align}s(\lfloor n'/\ln^2 n'\rfloor)
\le 
\frac{\tilde{T}_{\calD}(n')\cdot 40\ln n'}{n'}
\le 
\frac{o(n'\log^{\beta/2-1} n')\cdot 40\ln n'}{n'}=o(\log^{\beta/2} n').\label{eq5:high-slots}\end{align}

By~\Cref{lem:basic-filter}, we can upper bound the number of high-probability time slots by
\begin{align*}N^{\high}(\lceil n\log^{\beta/2-1} n\rceil)
&\le N^{\high}(\lfloor n'/\ln^2 n'\rfloor)
\le \frac{s(\lfloor n'/\ln^2 n'\rfloor)}{B_{\beta}(\lfloor n'/\ln^2 n'\rfloor)}\\
\intertext{According to the definition of $B_{\beta}(\cdot)$ and~\Cref{eq5:high-slots}, the above can be upper bounded by}
&\le o(\log^{\beta/2} n')\cdot\left(\frac{\lfloor n'/\ln^2 n'\rfloor}{\ln^{\beta} \lfloor n'/\ln^2 n'\rfloor}\right)\\
&=o(n'\log^{-\beta/2-2} n)
\;=\; o(n/\log n).
\end{align*}
\end{proof}

\subsection{Construction of the Adversary}
\label{subsec:construct-adversary}
In this subsection, we construct an adversary that maintains $\hat\sigma[t]=\Omega(\log\log n)$ over a longer interval,
whose length depends on
$s^{\mathrm{low}_\beta}$. 
At this level of contention there will be \emph{some} successes, so to bound the number of successes we need to prove that $\sigma[t]$ is also $\Omega(\log\log n)$, despite successful parties dropping out of the system and reducing the contention on future time slots.

The following general lemma shows how to construct an adversary that maintains $\hat\sigma[t]\ge \ell$ in a long interval and relate the length of the interval to $s^{\filter }$.

\begin{lemma}\label{lem:general-adversary}
    We are given a protocol $\mathcal{D}$, which defines $\eta\bydef p(1)>0$, a sufficiently large $n\ge n_0(\eta)$, 
    an integer $T_0 \in [1,n]$, a real parameter 
    $\ell\geq 1$, and a filter function $\filter$. Let $V=\max_{i\in [1,T_0]} p^{\filter}(i)$. If $e^{-\ell/(7V)}<1/n^2$, then there exists an oblivious adversary that:
    \begin{itemize}
        \item wakes up at most $\lfloor n/3\rfloor$ parties.
        \item $\hat\sigma^{\filter}[t]\ge \ell$ for all $t\in [T_0,T_1]$, where $T_1\bydef\left\lfloor \dfrac{n s^{\filter}(T_0)}{8\ell}\right\rfloor$. 
    \end{itemize}
\end{lemma}

\begin{proof}
    We independently assign the activation time of the $\lfloor n/3\rfloor$ parties to slots in $[0,T_1-1]$ uniformly at random,
    with replacement. For any time $t\in [T_0,T_1]$. The expected filtered aggregate contention at time $t$ is
    
$$
\E[\hat\sigma^{\filter}[t]]=\sum_{u=1}^{\lfloor n/3\rfloor}\frac{1}{T_1}\sum_{t_u=0}^{t-1}p^{\filter}(t-t_u)\geq \frac{\lfloor n/3\rfloor}{ns^{\filter}(T_0)/(8l)}\cdot s^{\filter}(T_0)\geq 2\ell. 
$$

\noindent Since the parties are activated independently, we apply a Chernoff bound (\Cref{coro:chernoff-bound-2}) to get

$$
    \Pr[\hat\sigma^{\filter}[t]<\ell]\leq e^{-\ell/(7V)}< 1/n^2.
$$
    
\noindent By a union bound over $[T_0,T_1]$,
$$
\Pr\left[\exists t\in[T_0,T_1] \text{ such that } \hat\sigma^{\filter}[t]<l\right]\le n^{-2}\left\lfloor \dfrac{n s^{\filter}(T_0)}{8\ell}\right\rfloor \le n^{-2}\cdot \frac{nT_0}{8}<1.
$$
Thus, there exists an oblivious adversary such that $\hat\sigma^{\filter}[t]\geq \ell$ for all $t\in[T_0,T_1]$. 
\end{proof}

Now assume that we have a lower bound for $s^{\low}(n)-s^{\low}(\sqrt{n})$ for sufficiently large $n$. Consider the filter $\filter=\low^{(\ge \sqrt{n})}$, where $\low^{(\ge \sqrt{n})}(i) = \low(i)\cdot \mathbf{1}\{i\ge \sqrt{n}\}$. We can construct an adversary as follows: 

\begin{theorem}\label{thm: high-contention construction}
    For any constants $\eta\in(0,1],
    \gamma>0,
    \beta\ge 10$ and acknowledgment-based \ContentionResolution{} protocol $\mathcal D$ satisfying $p(1)=\eta$, there exists an oblivious adversary such that for sufficiently large $n\ge n_0(\eta,\beta,\gamma)$, 
    $n_0$ depending only on $\eta,\beta$, and $\gamma$, 
    
    \begin{itemize}
        \item $\hat\sigma[t]\ge 10\ln n$ for all $t\in [1,T_0(n)]$,
        \item $\hat\sigma^{\mathrm{low}^{(\ge \sqrt{n})}_\beta}[t]\ge \gamma\ln\ln n$ for all $t\in [T_0(n), T_1(n)]$,
        \item The adversary wakes up at most $2n/3$ parties in $[1,T_1(n)]$,
    \end{itemize}
    
     where $T_0,T_1$ are defined to be
     \begin{align*} T_0(n)&\bydef\lfloor n/\ln^2 n\rfloor\\
     T_1(n)&\bydef\left\lfloor \dfrac{ns^{\mathrm{low}_\beta^{(\ge \sqrt{n})}}(T_0(n))}{8\gamma\ln\ln n}\right\rfloor=\left\lfloor \dfrac{n(s^{\mathrm{low}_\beta}(\lfloor n/\ln^2 n\rfloor)-s^{\mathrm{low}_\beta}(\lfloor\sqrt{n}\rfloor))}{8\gamma\ln\ln n}\right\rfloor.
     \end{align*}

\end{theorem}

\begin{proof}
We assign $\lfloor n/3\rfloor$ parties in $[0,T_0]$ as in~\Cref{lem:simple-adversary}. Let $\filter=\low^{(\ge \sqrt{n})}$ and $\ell =\gamma \ln\ln n$. Since $V=\max_{t\in[1,T_0]} p^{\low^{(\ge \sqrt{n})}}(t)=B_{\beta}(\lfloor \sqrt{n}\rfloor)$, $e^{-\ell/(7V)}=e^{-\Omega(\sqrt{n}/\polylog(n))}$,
which is less than $1/n^2$ 
when $n$ is sufficiently large,
as a function of $\beta,\gamma$.
Therefore, we can apply~\Cref{lem:general-adversary} in $[T_0,T_1]$ with $\filter=\low^{(\ge \sqrt{n})}$ and $\ell =\gamma \ln\ln n$.

Combining the above two results yields
the adversary satisfying the desired conditions.
\end{proof}

Given error probability function $q(n)$, which is either
$q(n)= 1/2$ or $q(n)= n^{-2}$, 
we can prove that the adversary of~\Cref{thm: high-contention construction} 
guarantees a low success probability for a specific party $u^*$ in $\left[ 1,T_1(n)\right]$ for appropriate choice of $\beta$ and $\gamma$, and therefore provides a lower bound for the \emph{restricted time window} $\tilde{T}_{\calD}(n)$. In particular, we have the following theorem.
\begin{restatable}{theorem}{RestrictedWindowLowerBound}\label{thm:restricted-window-lower-bound}
Given a memoryless protocol $\calD$, error probability $q(n)=1/2$ or $q(n)=n^{-2}$,
and constants $\beta\ge 10,\eta\bydef p(1)>0$, 
if 
$L_{\calD}(n,q(n))=o(n\log^{\beta/2-1} n)$, then for all sufficiently large $n$,
$$
\tilde{T}_{\calD}(n)\ge \left\lfloor \dfrac{n(s^{\mathrm{low}_\beta}(\lfloor n/\ln^2 n\rfloor)-s^{\mathrm{low}_\beta}(\lfloor\sqrt{n}\rfloor))}{40\beta\ln\ln n}\right\rfloor.$$
\end{restatable}
In the interest of readability we shall postpone the proof of \Cref{thm:restricted-window-lower-bound}
to~\Cref{subsec:proofofkey}. 

\subsection{Analysis for Lower Bounds}\label{subsec:analysis-for-lower-bounds}

In this section 
we apply~\Cref{thm:restricted-window-lower-bound} to prove \Cref{thm: main lower bound without global clock}. 
Fix a memoryless protocol $\calD$, and recall that $L_{\calD}^{\mathrm{whp}}(n)=L_{\calD}(n,n^{-2})$ and $L_{\calD}^{\mathrm{exp}}(n)\le 2L_{\calD}(n,1/2)$. 
We therefore only consider 
$q(n)=1/2$ when bounding expected latency and
$q(n)=1/n^2$
when bounding high-probability latency. 
For the sake of contradiction, we suppose that $L_\calD(n,q(n))=o(U(n))$, 
where $U(n)=\lfloor \frac{n\log(1/q(n))\log n}{\log\log n}\rfloor$, which
captures both claims of \Cref{thm: main lower bound without global clock}. 
When $n$ is sufficiently large we can avoid asymptotics and proceed under the assumption that $L_\calD(n,q(n)) < U(n)$.
By~\Cref{thm:restricted-window-lower-bound}, we know that
\begin{align}
\tilde{T}_{\calD}(n)&\ge\left\lfloor \dfrac{n(s^{\mathrm{low}_\beta}(\lfloor n/\ln^2 n\rfloor)-s^{\mathrm{low}_\beta}(\lfloor\sqrt{n}\rfloor))}{40\beta\ln\ln n}\right\rfloor.\label{eqn:hatT}
\end{align}

For memoryless protocols, we show that:
\begin{align}
s^{\mathrm{low}_\beta}(L_{\calD}(n,q(n)))-s^{\mathrm{low}_\beta}(\tilde{T}_{\calD}(n))
&\ge 
\frac{1}{4}\ln(1/q(n)).\label{eqn:s-low}
\end{align}

We consider the quantity $C(n) \bydef \tilde{T}(n)/U(n)$. 
By assumption, $C(n) = o(1)$. 
By combining \Cref{eqn:hatT,eqn:s-low}, we obtain (informally) that 
$C(n) \approx 1 / \ln(2/C(n))$, which leads to a contradiction.

\begin{lemma}\label{lem: slow lower bound for memoryless protocols}
Fix a memoryless protocol $\calD$, an error probability $q$, which is either $q(n)=1/2$ or 
$q(n) = n^{-2}$,
and constants $\beta\ge 10, \eta\bydef p(1) >0$.
If 
$L_{\calD}(n,q(n))=o(n\log^{\beta/2-1}n)$ then for $n$ sufficiently large, 
$$
s^{\mathrm{low}_\beta}(L_{\calD}(n,q(n)))-s^{\mathrm{low}_\beta}(\tilde{T}_{\calD}(n))\ge \frac{1}{4}\ln(1/q(n)).
$$
\end{lemma}
\begin{proof}
Consider a specific party $u^*$ activated at time $0$. 
We will derive a lower bound of the probability that $u^*$ does not succeed in $[1,L_{\calD}(n,q(n))]$ and compare it with $q(n)$. We may assume that $u^*$ does not exit after success since this does not change the probability that $u^*$ does not succeed in $[1,L_{\calD}(n,q(n))]$.

Let $E_1$ be the event that $u^*$ does not succeed in $[1,\tilde T(n)]$. By the definition of $\tilde{T}_{\calD}(n)$, there exists an adversary waking up at most $2n/3$ parties that guarantees $\Pr[E_1]\ge 4^{-1/8}$. 

Next, for each high-probability time slot $t\in (\tilde{T}(n),L_{\calD}(n,q(n))]$ for which $p(t)>B_\beta(t)$, 
we wake up $\lceil 4\ln n/\eta\rceil$ parties at global time $t-1$. 
Let $E_2$ be the probability that $u^*$ does not succeed in any of these high-probability time slots. This step wakes up at most $N^{\mathrm{high}}(L_{\calD}(n,q(n)))\lceil 4\ln n/\eta\rceil$ parties, guaranteeing that $\Pr[E_2]\ge 1-1/n\ge 4^{-1/8}$ when $n$ is large. 
By~\Cref{lem:restricted-window}(3),
$
N^{\high}(L_{\calD}(n,q(n)))\le N^{\high}(\lceil n\log^{\beta/2-1} n\rceil)=o(n/\ln n).
$
Therefore, when $n$ is sufficiently large, at most $n/3$ parties are activated in this step. Thus, at most $n$ parties are activated in total.  

Finally, let $E_3$ be the probability that $u^*$ does not succeed in any of the remaining slots in $(\tilde{T}_{\calD}(n),L_{\calD}(n,q(n))]$. Since $p^{\mathrm{low}_\beta}(t)\le B_{\beta}(t)=\frac{\ln^{\beta}t}{t}$ and $\tilde{T}_{\calD}(n)\ge\Omega(n/\ln^2 n)$ by~\Cref{lem:restricted-window}, we may assume $p^{\mathrm{low}_\beta}(t)\le 1/2$ for all $\tilde{T}_{\calD}(n)<t\le L_{\calD}(n,q(n))$ when $n$ is sufficiently large. Using the approximation $1-x\ge 4^{-x}$ for $x\in[0,1/2]$, we have

$$
\Pr[E_3]=\prod_{t\in (\tilde{T}_{\calD}(n), L_{\calD}(n,q(n))]}(1-p^{\mathrm{low}_\beta}(t))\ge (1/4)^{s^{\mathrm{low}_\beta}(L_{\calD}(n,q(n))-s^{\mathrm{low}_\beta}(\tilde{T}_{\calD}(n)))}.
$$

Since the protocol is memoryless, $E_1,E_2,E_3$ are independent. 
Thus, the probability that $u^*$ does not succeed in $[1,L_{\calD}(n,q(n))]$ is $\Pr[E_1]\Pr[E_2]\Pr[E_3]\ge 4^{-(s^{\mathrm{low}_\beta}(L_{\calD}(n,q(n)))-s^{\mathrm{low}_\beta}(\tilde{T}_{\calD}(n))+1/4)}$. Since $u^*$ succeeds before $L_{\calD}(n,q(n))$ with probability at least $1-q(n)$, we must have
$$
4^{-(s^{\mathrm{low}_\beta}(L_{\calD}(n,q(n)))-s^{\mathrm{low}_\beta}(\tilde{T}_{\calD}(n))+1/4)}\le q(n).
$$
It follows that
$$
s^{\mathrm{low}_\beta}(L_{\calD}(n,q(n)))-s^{\mathrm{low}_\beta}(\tilde{T}_{\calD}(n))\ge \frac{\ln(1/q(n))}{\ln 4}-1/4\ge \frac{1}{4}\ln(1/q(n)).
$$
\end{proof}

The following lemma summarizes the statements we need for the final lower bound.

\begin{lemma}\label{lem:n0}
    Given a error probability function $q(n)$ that is either identically $=1/2$ or identically $n^{-2}$, let $\beta=10$ and $U(n)=\left\lfloor \frac{n\ln(1/q(n))\ln n}{\ln\ln n}\right\rfloor$. If $L_{\calD}(n,q(n))=o(U(n))$, then there exists an integer $n_0$ such that the following statements hold when $n\ge n_0$:  
\begin{enumerate}[(i)]
    \item $\tilde{T}_{\calD}(n)\ge \left\lfloor \dfrac{n(s^{\mathrm{low}_\beta}(\lfloor n/\ln^2 n\rfloor)-s^{\mathrm{low}_\beta}(\lfloor\sqrt{n}\rfloor))}{40\beta\ln\ln n}\right\rfloor$.
    \item $\tilde{T}_{\calD}(n)\ge n/\ln^2 n$. 
    \item $L_\calD(n,q(n))>\tilde{T}_{\calD}(n)$.
    \item $U(n)>L_\calD(n,q(n))$.
    \item $\max\{U(n'):U(n')\le cU(n)\}\ge cU(n)/2$ whenever $20/n\le c<1$.
    \item $\displaystyle s^{\mathrm{low}_\beta}(L_{\calD}(n,q(n)))-s^{\mathrm{low}_\beta}(\tilde{T}_{\calD}(n))\ge \frac{1}{4}\ln(1/q(n))$.
    \item $U(n)\le\lfloor n\ln^2 n\rfloor$.
    \item $\ln(1/q(\sqrt{n}))\ge \frac{1}{3}\ln (1/q(n))$.
    \item $\frac{\ln\ln i}{\ln(1/q(i))\ln^3 i}\ge \frac{\ln\ln n}{\ln(1/q(n))\ln^3 n}\ge 1/\ln^4 (\sqrt{n})\ge 20/\sqrt{n}$ for all $i\le n$.
    \item $U(i)\le U(n)$ for all $1\le i\le n$.
    \item $\ln (\lfloor n/\ln^4 n \rfloor/(\sqrt{n}\ln^4 \sqrt{n}))\ge \frac{1}{5}\ln n$.
\end{enumerate}
\end{lemma}
\begin{proof}
Items~(i)–(iii) and~(vi) follow from 
\Cref{thm:restricted-window-lower-bound}, 
\Cref{lem:restricted-window}~(ii), 
\Cref{def:restricted-window},  
and~\Cref{lem: slow lower bound for memoryless protocols}, respectively. 
See Appendix~\ref{sect:smoothness-proofs} for technical lemmas
that imply Item~(v).
Item~(iv) is a consequence of the assumption that $L_{\calD}(n,q(n)) = o(U(n))$. 
Items~(vii)–(xii) are straightforward mathematical facts.
\end{proof}

We are now ready to prove \Cref{thm: main lower bound without global clock}.  To prove both the In-Expectation and With-High-Probability lower bounds of \Cref{thm: main lower bound without global clock}, it suffices to show that no protocol
$\calD$ has latency $L_\calD(n,q(n)) = o(\frac{n\log(1/q(n))\log n}{\log\log n})$,
where the error function is either 
$q(n)= 1/2$ or $q(n)= 1/n^2$.

\begin{proof}[Proof of \Cref{thm: main lower bound without global clock}]
Take $\beta=10$, and recall that $U(n)=\left\lfloor \frac{n\ln(1/q(n))\ln n}{\ln\ln n}\right\rfloor$.
Suppose, for the sake of obtaining a contradiction, that $L_\calD(n,q(n))=o(U(n))$. 
We choose $n_0$ as in~\Cref{lem:n0}
and henceforth assume 
$n\ge n_0^2$.

Let $C(n)=\tilde{T}_{\calD}(n)/U(n)$, which is
$o(1)$ since $\tilde{T}_{\calD}(n)< L_\calD(n,q(n))=o(U(n))$. 
It follows that 
$\ln (2/C(n))=o(C(n)^{-1})$.
By~\Cref{lem:n0}(ii), $\tilde{T}_{\calD}(n)\ge n/\ln^2 n$ and we also have $C(n)\ge \frac{\ln\ln n}{\ln(1/q(n))\ln^3 n}$. 
We define numbers $N_0,C_0,N$ as follows.
\begin{itemize}
    \item Let $N_0\ge n_0$ be a constant such that $2400\beta\ln(2/C(n))\le (2C(n))^{-1}$ for all $n\ge N_0$. The existence of $N_0$ is due to our assumption $C(n)=o(1)$ and the fact that $\ln(1/x)=o(1/x)$ for small $x$.
    \item Let $C_0=\min\limits_{N_0\le n\le N_0^2}C(n)$. By the definition of $N_0$ we have $2400\beta\ln(2/C_0)\le (2C_0)^{-1}$. 
    \item Let $N=\min\lbrace N':N'\ge N_0,C(N')<C_0\rbrace$ be the smallest integer that is not smaller than $N_0$ and satisfies $C(N)< C_0$. The existence of $N$ is due to our assumption $C(n)=o(1)$.
\end{itemize}

We prove that $C(N)\ge C_0$, which contradicts the definition of $N$. 
By the definition of $C_0$ and $N$, $N>N_0^2$ and $C(n)\ge C_0$ for all $\sqrt{N}\le n<N$.
Define the sequence $(N_i)$ by $N_1=\lfloor N/\ln^4 N\rfloor$ and for $i>1$, $N_i=\arg\max_{N'}\{U(N')\le C_0U(N_{i-1})\}$.  We terminate the sequence 
at $i_{\max}$, defined to be
$$
i_{\max}=\min\{i:\min(C_0U(N_i),N_{i})\le \sqrt{N}\}.
$$
By the definition of $C_0$, we have:
\begin{align}
1>C_0 &=\min_{N_0\le n\le N_0^2} C(n)= \min_{N_0\le n\le N_0^2}\left(\frac{\tilde{T}(n)}{U(n)}\right)\notag\\
\intertext{By~\Cref{lem:n0}(ii) this is lower bounded by}
&\ge \min_{N_0\le n\le N_0^2} \left(\frac{n/\ln^2 n}{n\ln(1/q(n))\cdot \ln n/\ln\ln n}\right)\notag\\ 
\intertext{Applying~\Cref{lem:n0}(ix),}
&\ge \frac{\ln\ln (N_0^2)}{\ln(1/q(N_0^2))\cdot \ln^3(N_0^2)}\ge \frac{1}{\ln^4 N_0}\notag\\
\intertext{Since $N\ge N_0^2$, we have:}
C_0&\ge \frac{1}{\ln^4 (\sqrt{N})} \label{eq5:c0lowerbound}
\end{align}
Since $\frac{1}{\ln^4 \sqrt{N}}\ge \frac{20}{\sqrt{N}}\ge \frac{20}{N_{i-1}}$ by~\Cref{lem:n0} (ix), applying~\Cref{lem:n0}(v), we have $U(N_i)\ge C_0U(N_{i-1})/2$ for $2\le i\le i_{\max}$, and by the definition of $i_{\max}$ we have
\begin{align*}
\min(C_0U(N_{i_{\max}}),N_{i_{\max}}) &\le \sqrt{N}.
\intertext{If this last inequality is satisfied because $C_0U(N_{i_{\max}})\le \sqrt{N}$,
then by~\Cref{eq5:c0lowerbound}, we have:}
U(N_{i_{\max}}) &\le \sqrt{N}/C_0 \le \sqrt{N}\ln^4 \sqrt{N}.
\intertext{If it were satisfied because $N_{i_{\max}}\le \sqrt{N}$, then by~\Cref{lem:n0}(x) and the definition of $U(\cdot)$,}
U(N_{i_{\max}}) &\le U(\sqrt{N})
\le \sqrt{N}\ln^4 \sqrt{N}.
\intertext{Therefore, in every case $U(N_{i_{\max}})\le \sqrt{N}\ln^4 \sqrt{N}$. By~\Cref{lem:n0}(xi), }
i_{\max}-1&\ge \frac{\ln (U(N_1)/U(N_{i_{\max}}))}{\ln(\max_{1\le i\le i_{\max}-1}U(N_{i+1})/U(N_i))}\\
&\ge \frac{\ln (\lfloor N/\ln^4 N \rfloor/(\sqrt{N}\ln^4 \sqrt{N}))}{\ln(2/C_0)}\\
&\ge \frac{\ln N}{5\ln(2/C_0)}.
\end{align*}

For $i\in [1,i_{\max}]$, since $n_0\le N_0\le \sqrt{N}$, we have $\tilde{T}_{\calD}(N_i)=C(N_i)\cdot U(N_i)\ge C_0U(N_i)$. 
By~\Cref{lem:n0}(vi,viii), we have:

\begin{align*}
s^{\mathrm{low}_\beta}(U(N_i))-s^{\mathrm{low}_\beta}(C_0U(N_i))&\ge s^{\mathrm{low}_\beta}(L_{\calD}(N_i,q(N_i)))-s^{\mathrm{low}_\beta}(\tilde{T}_{\calD}(N_i))\\
&\ge \frac{1}{4}\ln(1/q(N_i))\ge\frac{1}{4}\ln(1/q(\lfloor\sqrt{N}\rfloor))\ge \frac{1}{12}\ln(1/q(N)).
\intertext{Since the intervals $(C_0U(N_i),U(N_i)]$ are disjoint and are in $(\lfloor\sqrt{N}\rfloor,\lfloor N/\ln^2 N\rfloor]$, we have}
s^{\mathrm{low}_\beta}(\lfloor N/\ln^2 N\rfloor)-s^{\mathrm{low}_\beta}(\lfloor\sqrt{N}\rfloor)&\ge\sum_{i=1}^{i_{\max}-1} \left[\istrut[1]{4}s^{\mathrm{low}_\beta}(U(N_i))-s^{\mathrm{low}_\beta}(C_0U(N_i))\right]\\
&\ge (i_{\max}-1)\cdot \frac{1}{12}\ln(1/q(N))\ge \frac{\ln N}{5\ln(2/C_0)}\cdot \frac{1}{12}\ln(1/q(N))\\
&\ge\frac{U(N)\ln\ln N}{60\ln(2/C_0)N}.
\intertext{It follows that}
\tilde{T}_{\calD}(N)&\ge \left\lfloor \dfrac{n(s^{\mathrm{low}_\beta}(\lfloor N/\ln^2 N\rfloor)-s^{\mathrm{low}_\beta}(\lfloor\sqrt{N}\rfloor))}{40\beta\ln\ln N}\right\rfloor\\
&\ge \lfloor \frac{U(N)}{2400\beta\ln(2/C_0)}\rfloor\ge \lfloor 2C_0U(N)\rfloor\\
&\ge C_0U(N).
\end{align*}
The third inequality follows from the condition $2400\beta \ln(2/C_0) \le (2C_0)^{-1}$, as specified in the definition of $C_0$.
Therefore, $C(N)=\tilde{T}_{\calD}(N)/U(N)\ge C_0$, contradicting our choice of $N$.
\end{proof}


\subsection{Tradeoff Between Expected Latency and High-probability Latency}\label{subsec:tradeoff}

By applying an argument that is similar to~\Cref{thm: main lower bound without global clock}, we can prove the tradeoff lower bound
of \Cref{thm: tradeoff between exp-latency and whp-latency without a global clock}, 
which states that no memoryless protocol
$\calD$ has both 
$L_\calD^{\mathrm{exp}}(n)=o(\frac{n\log^2 n}{(\log\log n)^2})$ 
and 
$L_\calD^{\mathrm{whp}}(n)=n\log^{O(1)}n$.

\begin{proof}[Proof of \Cref{thm: tradeoff between exp-latency and whp-latency without a global clock}]
    Suppose $L_\calD^{\mathrm{whp}}(n)\le \bar{U}(n)=n\ln^{\beta/2-1} n$ for some constant $\beta\ge 10$.

    By~\Cref{lem:restricted-window}(2), we know $\tilde{T}_{\calD}(n)\ge n/\ln^2 n$ for sufficiently large $n$.
    By~\Cref{lem: slow lower bound for memoryless protocols}, 
    \begin{align*}s^{\low}(\lfloor n\ln^{\beta/2-1} n\rfloor)-s^{\low}(\lfloor n/\ln^2 n\rfloor)&\ge s^{\low}(L_{\calD}(n,n^{-2}))-s^{\low}(\tilde{T}_{\calD}(n))\\
    &\ge \frac{1}{4} \ln (1/(n^{-2}))\ge \frac{1}{2}\ln n.
    \end{align*}

    We define $\bar{C_0}=\min_{\sqrt{n}\le t\le n} \left( \frac{\tilde{T}_{\calD}(n)}{\bar{U}(n)}\right)$, then we have
    \begin{align}\bar{C_0}\ge \min_{\sqrt{n}\le t\le n}\frac{\lfloor t/\ln^2 t\rfloor}{\lfloor t\ln^{\beta/2-1} t\rfloor}\ge \ln^{-\beta/2-2} n \quad\text{for sufficiently large $n$}\label{eq5:tradeoff-C0-lower-bound}\end{align}.
    We choose a series of thresholds $\lbrace n_i\rbrace_{1\le i\le i_{\max}}$ such that $n_1=\lfloor n/\ln^4 n\rfloor$, $n_i=\arg\max_{n'} \lbrace \bar{U}(n')\le \bar{C_0}\bar{U}(n_{i-1})\rbrace$ for $i\ge 2$, and
    $$i_{\max}=\min\lbrace i:\min(C(n)U(n_i),n_i)\le \sqrt{n}\rbrace.$$
    We can notice that:
    \begin{itemize}
    \item If $C_0U(n_{i_{\max}})\le \sqrt{n}$, then by~\Cref{eq5:tradeoff-C0-lower-bound}, $U(n_{i_{\max}})\le \sqrt{n}/\bar{C_0}\le \sqrt{n}\cdot \ln^{\beta/2+2} n$.
    \item Otherwise, $n_{i_{\max}}\le \sqrt{n}$, then $\bar{U}(n_{i_{\max}})\le \bar{U}(\lfloor \sqrt{n}\rfloor)\le \sqrt{n}\ln^{\beta/2+2} n$.
    \end{itemize}
    Thus we have: \begin{align}U(n_{i_{\max}})\le \sqrt{n}\ln^{\beta/2+2} n\quad \text{ for sufficiently large $n$}\label{eq5:tradeoff-u-upper-bound}\end{align} 
    
For sufficiently large $n$ and $2\le i\le i_{\max}$, by the fact  $\bar{U}(n_i+1)/\bar{U}(n_{i-1})>\bar{C_0}$ and~\Cref{eq5:tradeoff-C0-lower-bound}, we can notice that: \begin{equation}
    \begin{aligned}\bar{U}(n_i)/\bar{U}(n_{i-1})&=\left(\frac{\bar{U}(n_i+1)}{\bar{U}(n_i+1)}\right)-\left(\frac{\bar{U}(n_i+1)-\bar{U}(n_{i})}{\bar{U}(n_{i-1})}\right)\\
    &\ge \bar{C_0}-O(\ln^{\beta/2-1} n)/\Omega(\sqrt{n}\cdot \ln^{\beta/2+2} n)\\
    &\ge \bar{C_0}-O(1/\sqrt{n})\\
    &\ge \bar{C_0}/2
\end{aligned}
\label{eq5:tradeoff-ui/ui-1}
\end{equation}
and by~\Cref{eq5:tradeoff-u-upper-bound} and~\Cref{eq5:tradeoff-ui/ui-1}, 
\begin{align*}
    i_{\max}-1&\ge \frac{\ln(\bar{U}(n_1)/\bar{U}(n_{i_{\max}}))}{\ln(\max_{2\le i\le i_{\max}} \bar{U}(n_{i-1})/\bar{U}(n_{i}))}\\
    &\ge \frac{\ln (\lfloor n/\ln^4 n\rfloor/(\sqrt{n}\ln^{\beta/2+2} n))}{\ln(2/\bar{C_0})}\\
    &\ge \Omega\left(\frac{\log n}{\log \log n}\right).
\end{align*}
Lastly, by the definition of $\bar{C_0}$, we know $\bar{C_0}\bar{U}(n_i)\le \bar{T}_{\calD}(n)$ so $(\bar{C_0} \bar{U}(n_i),\bar{U}(n_i)]\,(1\le i<i_{\max})$  are disjoint intervals. Therefore, 
\begin{align*}
    s^{\low}(\lfloor n/\ln^2 n\rfloor)-s^{\low}(\lfloor \sqrt{n}\rfloor)&\ge \sum_{i=1}^{i_{\max}-1} \left[\istrut[1]{4}s^{\low}(U(n_i))-s^{\low}(\bar{C_0}\cdot \bar{U}(N_i))\right]\\
    &\ge (i_{\max}-1)\cdot \frac{1}{2} \ln n_i\ge \Omega\left(\frac{\log^2 n}{\log\log n}\right).
    \end{align*}
    
    \noindent Applying~\Cref{thm:restricted-window-lower-bound}, 
    $$L^{\mathrm{exp}}_{\calD}(n)\ge \tilde{T}_{\calD}(n)\ge \left\lfloor \frac{n(s^{\low}(\lfloor n/\ln^2 n\rfloor)-s^{\low}(\lfloor \sqrt{n}\rfloor))}{40\beta \ln\ln n}\right\rfloor=\Omega\left(\frac{n\log^2 n}{(\log\log n)^2}\right).$$
\end{proof}

\subsection{Proof of\texorpdfstring{~\Cref{thm:restricted-window-lower-bound}}{ Theorem 5.8}}\label{subsec:proofofkey}

Recall that \Cref{thm: high-contention construction} gave an adversary that sustains a high \emph{static} 
contention $\hat{\sigma}^{\filter}[t] = \Omega(\log\log n)$ 
over a certain interval of time, 
for the filter $\filter = \low^{(\geq\sqrt{n})}$.  
This does not \emph{directly} 
imply a similar
lower bound $\sigma^{\filter}[t]=\Omega(\log\log n)$, as parties achieving success \emph{reduce} the dynamic contention on future time slots.

The idea of the proof of \Cref{thm:restricted-window-lower-bound} 
is as follows. We divide the time slots into intervals of length $\delta=\log^{\Theta(1)} n$. If $\sigma^{\filter}[t]=\Omega(\log\log n)$ in an interval, then there will be few successes in that interval with very high probability. If there are few successes in previous intervals, then the difference between
$\sigma^{\filter}[t]$ and $\hat{\sigma}^{\filter}[t]$ caused by these successes is negligible. 
By using the filter $\filter=\mathrm{low}_\beta^{(\ge \sqrt{n})}$ that removes all high probability slots, we guarantee that
$\sigma^{\filter}[t]$ does not decrease much within an interval. Consequently, we have $\sigma^{\filter}[t]\approx \hat\sigma^{\filter}[t]=\Omega(\log\log n)$ in the next interval. By induction, we can prove $\sigma^{\filter}[t]=\Omega(\log\log n)$ for all $t\in[1,T]$ with high probability, provided that $\hat\sigma^{\filter}[t]=\Omega(\log\log n)$ initially,
for all $t\in[1,T]$. Then, the expected number of successes of a specific party $u^*$ activated at time $0$ is at most $s(T)/e^{\Omega(\log\log n)}=\log^{-\Omega(1)}n$ for appropriately chosen constants. 
It follows that the success probability of $u^*$ 
is at most $\log^{-\Omega(1)}n$, 
which yields a lower bound $T\leq \tilde{T}_{\calD}(n)$.

\medskip 

We begin with some definitions specific to this section.
\begin{definition}[Success Record, $(\mu,I)$-Goodness, Density Profiles] \ 
\begin{enumerate}
\item Define $\mathrm{Succ}[t] \bydef \hat A[t]\setminus A[t]$ as the set of parties that have succeeded by global time $t$. Recall that $A[t\mid t']$ is the set of parties that are activated by time $t$ and yet to succeed by time $t'$. By this definition, we have 
$A[t \mid t'] = \hat{A}[t]\setminus \mathrm{Succ}[t']$, for $t'\le t$.
\item Define the \emph{success record} by global time $t$ as $\mathrm{Rec}[t]=\{(u,t^{\text{succ}}_u):u\in \mathrm{Succ}[t]\}$, i.e., the set of successful parties along with their success time by global time $t$. By abuse of notation, we write $u\in \mathrm{Rec}[t]$ to mean that there exists a pair $(u,\tau)\in \mathrm{Rec}[t]$ for some $\tau$.
\item A success record $R$ is said to be \emph{$(\mu,I)$-good} for a real number $\mu\in[0,1]$ and a set of time slots $I$ if the number of successes in $I$ is at most $\mu|I|$. 
That is, the success \emph{density} in $I$ is low. 
\item
A \emph{density profile} $\Pi$ is a collection of pairs $(\mu,I)$, where $\mu\in [0,1]$ is a real number and $I$ is an interval of time slots. A success record $R$ is said to be \emph{$\Pi$-good} if it is $(\mu,I)$ good for every $(\mu,I)\in\Pi$.
\item For integers $t_0,\delta\ge 1$ and real number $\mu\in [0,1]$, the \emph{$(t_0,\mu,\delta)$-density profile} $\Pi(t_0,\mu, \delta)$ is defined as follows: 
\begin{itemize}
    \item $I_0=[1,t_0)$ with density $\mu_0=0$;
    \item For $i\ge 1$, let $I_i=[t_0+(i-1)\delta,t_0+i\delta)$ and $\mu_i=\mu$;
    \item Then $\Pi(t_0,\mu,\delta)=\{(\mu_i,I_i)\}$. 
\end{itemize}
That is, we divide the time slots into intervals of length $\delta$ except the first interval $[1,t_0)$, and require no success in the first interval and at most $\mu\delta$ successes in other intervals.
\item A success record $R$ is said to be \emph{$(t_0,\mu,\delta)$-good} if it is $\Pi(t_0,\mu,\delta)$-good. 
\end{enumerate}
\end{definition}




\medskip 
We restate our plan as follows:

\begin{itemize}
    \item In~\Cref{subsubsec:process-analysis-1}, we show that if $\sigma^{\filter}[t; A[t\mid t_0]]=\Omega(\log\log n)$ for all $t\in [t_0,t_0+\delta)$, then the success record is $(2\mu,I)$-good with high probability where $\mu\approx \log^{-\Omega(1)}n$ and $I=[t_0,t_0+\delta)$.
    \item \Cref{subsubsec:process-analysis-2} shows that if the success record is $(\mu_i,I_i)$-good for all $i<k$ where $(\mu_i,I_i)\in\Pi(t_0,2\mu,\delta)$, then it is also $(\mu_k,I_k)$-good with high probability, provided that $\hat\sigma^{\filter}[t]=\Omega(\log\log n)$ in $I_k$. By induction, an adversary satisfying the properties in \Cref{thm: high-contention construction} guarantees that the success record $\mathrm{Rec}[T]$ is $(t_0,2\mu,\delta)$-good with high probability, and hence $\sigma^{\filter}[t]=\Omega(\log\log n)$ for all $t\in [1,T]$ with high probability. 
    \item Finally, we show that the expected number of successes of a particular party $u^*$ is $\log^{-\Omega(1)} n$, thus concluding the lower bound of $\tilde T(n)$. Details are provided in~\Cref{subsubsec:process-analysis-3}.
\end{itemize}

The following fact is useful:

\begin{fact}\label{fact:trivial facts about sigma}
For any memoryless protocol $\calD$, filter $\filter$, and time $t_0\le t$,
$$
\sigma^{\filter}[t; A[t\mid t_0]]=\hat \sigma^{\filter}[t]-\sigma^{\filter}[t;\mathrm{Succ}[t_0]].
$$
\end{fact}


\begin{lemma}[Generalization of~\Cref{lem:prob-of-success}(2).]\label{fact: prob of success}
For any filter $\filter$ and time slot $t$,
$$
\begin{aligned}\Pr\left[\sum_{u\in A[t]}X_{t-t_u}^{(u)}= 1 \right]
\le (1+e\cdot \sigma^{\filter}[t])e^{-\sigma^{\filter}[t]}.
\end{aligned}
$$
\end{lemma}
\begin{proof}
    The proof is similar to \Cref{lem:prob-of-success}(2). Since $X_{t-t_u}^{(u)}\ge X^{(u)}_{t-t_u}\filter(t-t_u)$,
    \begin{align}
    \Pr\left[\sum_{u\in A[t]}X_{t-t_u}^{(u)}= 1 \right]\le \Pr\left[\sum_{u\in A[t]}X_{t-t_u}^{(u)}\filter(t-t_u)\le 1\right] \label{eq5:suc-prob-1}
    \end{align}
    Since $X^{(u)}_{t-t_u}\,(u\in A[t])$ are mutually independent, 
\begin{equation}
\begin{aligned}
\Pr\left[\sum_{u\in A[t]}X_{t-t_u}^{(u)}\filter(t-t_u)=0 \right]
&=\prod_{u\in A[t]}\left(1-\Pr\left[X_{t-t_u}^{(u)}\filter(t-t_u)=1\right]\right)\le e^{-\sigma^{\filter}[t]}
\end{aligned}
\label{eq5:succ-prob-2}
\end{equation}
and
\begin{align}
\lefteqn{\Pr\left[\sum_{u\in A[t]}X_{t-t_u}^{(u)}\filter(t-t_u)=1 \right]}\nonumber\\
&=\sum_{v\in A[t]}\Prob{X_{t-t_v}^{(v)}\filter(t-t_v)=1}\prod_{u\in A[t]\backslash\{v\}}\left(1-\Pr\left[X_{t-t_u}^{(u)}\filter(t-t_u)=1\right]\right)\nonumber\\
&\le \sum_{v\in A[t]}\Pr[X_{t-t_v}^{(v)}\filter(t-t_v)=1]e^{-\sigma^{\filter}[t]+1}=\sigma^{\filter}[t]e^{-\sigma^{\filter}[t]+1}.\label{eq5:succ-prob-3}
\end{align}
Combining~\Cref{eq5:suc-prob-1,eq5:succ-prob-2,eq5:succ-prob-3}, we have:
\begin{align*}
\Pr\left[\sum_{u\in A[t]}X_{t-t_u}^{(u)}\filter(t-t_u)\le 1 \right]
&=\Pr\left[\sum_{u\in A[t]}X_{t-t_u}^{(u)}\filter(t-t_u)=0 \right]+\Pr\left[\sum_{u\in A[t]}X_{t-t_u}^{(u)}\filter(t-t_u)= 1 \right]\\
&\le (1+e\cdot \sigma^{\filter}[t])e^{-\sigma^{\filter}[t]}.
\end{align*}
\end{proof}

\subsubsection{Proving \texorpdfstring{$(2\mu,I)$}{(2μ,I)}-goodness} \label{subsubsec:process-analysis-1}

If $\sigma[t;A[t\mid t_0]]$ is large for all $t \in [t_0, t_0 + \delta)$, then there are few successes in this interval. To prove this, we first introduce the following lemma, which shows that the loss in $\sigma^{\mathrm{low}_\beta^{(\ge \sqrt{n})}}[t]$ caused by successes in a short interval is small.

\begin{lemma}\label{lem: determinstic bound of sigma decrease}
    For any memoryless protocol $\calD$, integers $t_0,\delta,t$ with $t\ge t_0+\delta$,   
    
    $$
    \sigma^{\mathrm{low}^{(\ge \sqrt{n})}_\beta}[t;A[t\mid t_0+\delta]]\ge \sigma^{\mathrm{low_\beta^{(\ge \sqrt{n})}}}[t;A[t\mid t_0]]-B_\beta(\sqrt{n})\delta.
    $$
\end{lemma}

\begin{proof}
    Each party can contribute at most $B_\beta(\sqrt{n})$ to $ \sigma^{\mathrm{low}^{(\ge \sqrt{n})}_\beta}[t;A[t\mid t_0+\delta]]$, and there are at most $\delta$ parties that succeed in $[t_0,t_0+\delta)$. The statement follows.
\end{proof}

\begin{lemma}\label{lem: good success record in an interval}
    Fix a memoryless protocol $\calD$ and integers $t_0,\delta\ge 1,\ell\ge 2$. Suppose $\sigma^{\mathrm{low_\beta^{(\ge \sqrt{n})}}}[t;A[t\mid t_0]]\ge \ell$ for all $t\in [t_0,t_0+\delta)$ and $B_\beta(\sqrt{n})\delta\le \ell/2$. Then, conditioned on $H_{t_0}$, the success record is $(2\mu,I)$-good with probability at least $1-e^{-\mu\delta/3}$, where $\mu=(1+(e\cdot \ell/2))e^{-(\ell/2)})$ and $I=[t_0,t_0+\delta)$.
\end{lemma}

\begin{proof}
    By \Cref{lem: determinstic bound of sigma decrease} and our assumptions that 
    $\sigma^{\mathrm{low_\beta^{(\ge \sqrt{n})}}}[t;A[t\mid t_0]]\ge \ell$ and 
    $B_\beta(\sqrt{n})\delta\le \ell/2$,
    
    $$
    \begin{aligned}
    \sigma^{\mathrm{low_\beta^{(\ge \sqrt{n})}}}[t;A[t]]&\ge \sigma^{\mathrm{low_\beta^{(\ge \sqrt{n})}}}[t;A[t\mid t_0]]-(t-t_0)B_{\beta}(\sqrt{n})\\
    &\ge \sigma^{\mathrm{low_\beta^{(\ge \sqrt{n})}}}[t;A[t\mid t_0]]-\delta\cdot B_\beta(\sqrt{n})\\
    &\ge \ell/2. 
    \end{aligned}
    $$

    Then, by \Cref{fact: prob of success}, the success probability at time $t$ is at most $\mu=(1+e\cdot \ell/2)e^{-\ell/2}$. Hence, the expected number of successes in $[t_0,t_0+\delta)$ is at most $\mu\delta$. Moreover, since the probability bound is independent of the history within the interval $[t_0,t)$, we can apply the Chernoff bound (\Cref{coro:chernoff-bound-2}). It follows that the probability of more than $2\mu$ successes is at most $e^{-\mu\delta/3}$.
\end{proof}

\subsubsection{Proving \texorpdfstring{$(t_0,2\mu,\delta)$}{(t0,2μ,δ)}-goodness and High-contention} \label{subsubsec:process-analysis-2}

We now jointly establish, with high probability, that under the adversary constructed in \Cref{thm: high-contention construction}, the success record is $(t_0, 2\mu, \delta)$-good for appropriate parameters $t_0, \mu, \delta$, and that $\sigma^{\mathrm{low}_\beta^{(\ge \sqrt{n})}}[t] = \Omega(\log\log n)$ also holds. Let $I_i$ be the $i$th interval in the density profile $\Pi(t_0, \mu, \delta)$ (Recall that $I_0=[1,t_0)$ and $I_i=[t_0+(i-1)\delta,t_0+i\delta)$). We first show that if the success record is good, then the contention loss due to successes is small.

\begin{lemma}\label{lem: upper bound of sigma loss}
Given parameters $t'\le t,\delta\ge e^{\beta},\mu>0$, if $\mathrm{Rec}[t']$ is $(t_0,\mu,\delta)$-good, then
\[
\sigma^{\mathrm{low}_\beta^{(\ge \sqrt{n})}}[t; \mathrm{Succ}[t']] \le 3 B_\beta(\sqrt{n}) \delta + \frac{\mu \ln^{\beta+1} t}{\beta + 1}.
\]
\end{lemma}

\begin{proof}
By definition,

$$
\sigma^{\mathrm{low}_\beta^{(\ge \sqrt{n})}}[t;\mathrm{Succ}[t']]=\sum_{u\in \mathrm{Succ}[t']}p^{\mathrm{low}_\beta^{(\ge \sqrt{n})}}(t-t_u).
$$
    
    Suppose $t\in I_k$. If a party $u\in \mathrm{Succ}[t']$ succeeds at time $t_u^{\success}$, then $t_u<t_u^{\success}$ and $p^{\mathrm{low}_\beta^{(\ge \sqrt{n})}}(t-t_u)\le B_\beta(\max(t-t_u^{\success},\sqrt{n}))$. So the contribution of parties that succeeded in $I_i$ for each $i \le k-3$ is at most $\mu \delta \cdot B_\beta((k - i - 1)\delta)$. For $i \in \{k-2, k-1, k\}$, their total contribution is at most $3\delta \cdot B_\beta(\sqrt{n})$. So
    
    $$
    \begin{aligned}
        \sigma^{\mathrm{low}_\beta^{(\ge \sqrt{n})}}[t;\mathrm{Succ}[t']]&\le 3\delta\cdot B_\beta(\sqrt{n})+\mu\delta\sum_{i=2}^{k-1}\frac{\ln^{\beta}(i\delta)}{i\delta}\\
        &\le 3\delta\cdot B_\beta(\sqrt{n})+\mu\int_{\delta}^t\frac{\ln^{\beta}x}{x}\mathrm{d}x\\
        &\le 3\delta\cdot B_\beta(\sqrt{n})+\frac{\mu \ln^{\beta+1}t}{\beta+1}.
    \end{aligned}
    $$

\end{proof}

\begin{lemma}\label{lem: memoryless protocol main lemma}
    Let $n$ be sufficiently large, $t_0=\lfloor n/\ln^2 n\rfloor$, $\delta=\lfloor \ln^{2\beta+4}n\rfloor$ and $\mu=(1+e\cdot 2\beta\ln\ln n)e^{-2\beta\ln\ln n}$.  Suppose there exists an adversary $\calA[n]$ and an integer $T$ with $t_0<T\le n^2$ such that
    \begin{itemize}
        \item $\hat\sigma[t]\ge 10\ln n$ for all $t\in [1,t_0]$, and
        \item $\hat\sigma^{\mathrm{low}_\beta^{(\ge\sqrt{n})}}[t]\ge 5\beta\ln\ln n$ for all $t\in (t_0,T]$.
    \end{itemize}
    Then for any $k\ge 1$, 
    the following hold:

\begin{enumerate}
    \item Let $E_i$ be the event that the success record is $(2\mu, I_i)$-good. If events $E_0,E_1,\ldots, E_{k-1}$ all happen, then for any $t\in I_k\cap [1,T]$, $\sigma^{\mathrm{low_\beta^{(\ge \sqrt{n})}}}[t;A[t\mid t_0+(k-1)\delta]]\ge \sigma^{\mathrm{low_\beta^{(\ge \sqrt{n})}}}[t]\ge 4\beta\ln\ln n$.
    
    \item If $I_k\subseteq [t_0,T]$, then $\Pr[E_k\mid E_0\cap E_1\cap \cdots \cap E_{k-1}]\ge 1-e^{-\ln^3 n}$. As a consequence, $\Pr[E_0\cap E_1\cdots\cap E_k]\ge 1-n^{-6}-k\cdot e^{-\ln^3 n}\ge 1-2n^{-6}$, and $\Pr\left[\sigma^{\mathrm{low}_\beta^{(\sqrt{n})}}[t] \ge 4\beta\ln\ln n\right]\ge 1-2n^{-6}$ for all $t\in[t_0,T]$.
\end{enumerate}
\end{lemma}

\begin{proof}
Suppose $E_i$ happens for all $i<k$. Then $\mathrm{Rec}[t_0+(k-1)\delta]$ is $(t_0,2\mu,\delta)$-good. By \Cref{lem: upper bound of sigma loss}, \Cref{fact:trivial facts about sigma} and \Cref{lem: determinstic bound of sigma decrease}, for any $t\in I_k\cap [1,T]$,
\[
    \begin{aligned}
    \sigma^{\mathrm{low_\beta^{(\ge \sqrt{n})}}}[t;A[t]]&\ge\sigma^{\mathrm{low_\beta^{(\ge \sqrt{n})}}}[t;A[t\mid t_0+(k-1)\delta]]-B_\beta(\sqrt{n})\delta\\
    &\ge \hat\sigma^{\low^{(\ge\sqrt{n})}}[t]-\sigma^{\low^{(\ge\sqrt{n})}}\left[t;\mathrm{Succ}[t_0+(k-1)\delta]\right]-B_{\beta}(\sqrt{n})\delta\\
    &\ge \hat\sigma^{\mathrm{low}_\beta^{(\ge \sqrt{n})}}[t]-\left(3B_\beta(\sqrt{n})\delta+\frac{\mu\ln^{\beta+1}t}{\beta+1}\right)-B_\beta(\sqrt{n})\delta\\
     &\ge 5\beta\ln\ln n-4B_\beta(\sqrt{n})\delta-\frac{\mu\ln^{\beta+1}t}{\beta+1}.\\
    \end{aligned}
\]
For  sufficiently large $n$, we have $B_\beta(\sqrt{n})=\frac{\ln^{\beta}\sqrt{n}}{\sqrt{n}}$ and $\mu\le \ln^{-2\beta+1} n\le \ln^{-2\beta+1}n^2$. We also have $t\le T\le n^2$ by our assumption. Thus,
\[
\begin{aligned}
4B_\beta(\sqrt{n})\delta+\frac{\mu\ln^{\beta+1}t}{\beta+1}&\le\frac{4\ln^{\beta }\sqrt{n}}{\sqrt{n}}\ln^{2\beta+4}n+\frac{\ln^{2-\beta}n^2}{\beta+1}=o(1).
\end{aligned}
\]
So $\sigma^{\mathrm{low_\beta^{(\ge \sqrt{n})}}}[t;A[t\mid t_0+(k-1)\delta]]\ge \sigma^{\mathrm{low_\beta^{(\ge \sqrt{n})}}}[t;A[t]]\ge 4\beta\ln\ln n$.

By \Cref{lem: good success record in an interval}, if $\sigma^{\mathrm{low_\beta^{(\ge \sqrt{n})}}}[t;A[t\mid t_0+(k-1)\delta]]\ge 4\beta\ln\ln n$ for all $t\in I_k$, then $E_k$ happens with probability at least $1-e^{-\mu\delta/3}\ge 1-e^{-\ln^3 n}$ when $n$ is sufficiently large. So $\Pr[E_i\mid E_0\cap E_1\cap\cdots \cap E_{i-1}]\ge 1-e^{-\ln^3 n}$. 

Next, we prove by induction that $\Pr[E_0\cap E_1\cap \cdots\cap E_k]\ge 1-n^{-6}-k\cdot e^{-\ln^3 n}\ge 1-2n^{-6}$ when $n$ is large enough. The base case $\Pr[E_0]\ge 1-n^{-6}$ because of the fact that $\hat\sigma[t]\ge 10\ln n$ for $t\le t_0$ and~\Cref{lem:simple-adversary}. Then for $k\ge 1$,

$$
\begin{aligned}
\Pr[E_0\cap E_1\cap \cdots\cap E_k]&=\Pr[E_0\cap E_1\cap\cdots\cap E_{k-1}]\Pr[E_k\mid E_0\cap E_1\cap\cdots\cap E_{k-1}]\\
&\ge (1-n^{-6}-(k-1)e^{-\ln^3 n})(1-e^{-\ln^3 n})
\;\ge\; 1-n^{-6}-k\cdot e^{-\ln^3 n}.
\end{aligned}
$$
\end{proof}

\subsubsection{Proof of Latency Lower Bound} \label{subsubsec:process-analysis-3}

We are now prepared to restate and prove \Cref{thm:restricted-window-lower-bound}.

\RestrictedWindowLowerBound*
\begin{proof}
    Since $L_{\calD}(n,q(n))=o(n\log^{\beta/2-1}n)$, $s(\lfloor n/\ln^2 n\rfloor)\le \frac{40\ln n}{n}L_{\calD}(n,q(n))=o(\log^{\beta/2}n)$ by \Cref{lem:restricted-window}. It follows that $s(n^2)=o(\log^{\beta/2}n^2)=o(\log^{\beta/2} n)$.
    
    Now let $\calA$ be the adversary constructed by \Cref{thm: high-contention construction} with $\gamma=5\beta$ and
    $$
    T=\left\lfloor \dfrac{n(s^{\mathrm{low}_\beta}(\lfloor n/\ln^2 n\rfloor)-s^{\mathrm{low}_\beta}(\lfloor\sqrt{n}\rfloor))}{40\beta\ln\ln n}\right\rfloor\le n^2.
    $$
    Then $\calA$ and $T$ meet the requirements in \Cref{lem: memoryless protocol main lemma}. Consider a party $u^*$ activated at time $0$. For any time $t\in [1,T]$, if $\sigma^{\mathrm{low_\beta^{(\ge m)}}}[t]\ge 4\beta\ln\ln n$, then the probability that $u^*$ succeeds at time $t$ is at most $p(t)e^{-4\beta\ln\ln n}=\frac{p(t)}{\ln^{4\beta} n}$. By \Cref{lem: memoryless protocol main lemma}, $\sigma^{\mathrm{low_\beta^{(\ge m)}}}[t]\ge 4\beta\ln\ln n$ holds 
    with probability at least $1-2n^{-6}$ for all $t\in[\lfloor n/\ln^2 n\rfloor,T]$. So the probability that $u^*$ succeeds at time $t$ is at most $\frac{p(t)}{\ln^{4\beta}n}+2n^{-6}$ for all $t\in [\lfloor n/\ln^2 n\rfloor,T]$. For $[1,\lfloor n/\ln^2 n\rfloor)$, \Cref{lem:simple-adversary} implies that there are no successes in this interval with probability at least $1-n^{-8}$. Thus, the expected number of $u^*$'s successes in $[1,T]$ is at most    
    $$
    \frac{\lfloor  n/\ln^2 n\rfloor}{n^8}+\sum_{t=\lfloor n/\ln^2 n\rfloor}^{T}\left(\frac{p(t)}{\ln^{4\beta}n}+2n^{-6}\right)\le \frac{s(T)}{\ln^{4\beta}n}+2n^{-4}.
    $$
    By Markov's inequality, the probability that $u^*$ succeeds in $[1,T]$ is at most 
    $$
    \frac{s(T)}{\ln^{4\beta}n}+2n^{-4}\le\frac{s(n^2)}{\ln^{4\beta}n}+2n^{-4}=\frac{o(\ln^{\beta/2} n)}{\ln^{4\beta}n}+2n^{-4}=o(1).
    $$
    According to the definition of $\tilde{T}_{\calD}(n)$, we have $\tilde{T}_{\calD}(n)\ge T$ when $n$ is sufficiently large.
\end{proof}

\section{Upper Bounds in the \texorpdfstring{\LocalClock{}}{LocalClock} Model}\label{sect:upper-bounds}

In order to simplify the analysis of \ContentionResolution{} protocols, 
we reduce the problem to analyzing the length of a \emph{counter game}, 
for which the optimal adversarial strategy is obvious. 
In \Cref{sect:counter-games}, we bound the time of a counter 
game in terms of its parameters. In \Cref{sect:upper-bounds-exp-whp}, we apply this reduction
to bound the latency of \ContentionResolution{} protocols under
the metrics \emph{in expectation} and \emph{with high probability}. 

At a high level, we fix a special party $u^*$ and aim to evaluate its latency. For both expected latency and high-probability latency, we design a protocol and select an interval $\lbrack a, a + r)$, where $a \ge t_{u^*}$ and $r \ge n$. Our goal is to show that the probability that $u^*$ succeeds within this interval, conditioned on the event that it has not succeeded before time $a$, is either at least a constant (in the case of expected latency) or at least $1 - n^{-2}$ (in the case of high-probability latency). For a specific time slot $t \in \lbrack a, a + r)$, we focus on the quantity $\sigma[t] = \sum_{u \in A[t]} p(t - t_u)$, which represents the aggregate contention of the remaining parties at global time $t$. We consider the following three cases:

\begin{description}
    \item[Low Contention.] $\sigma[t] \le 1$: In this case, party $u^*$ succeeds with probability at least $p(t-t_{u^*})/4$, by \Cref{lem:prob-of-success}.
    \item[Medium Contention.] $1 < \sigma[t] \le \frac{1}{4} \log_2\log_2 n$: There is a non-negligible probability that at least one party succeeds.
    \item[High Contention.] $\sigma[t] > \frac{1}{4} \log_2\log_2 n$: This case cannot occur too many times, since the total sum of $\sigma[t]$ over the interval is bounded.
\end{description}

Rather than analyze the \emph{actual} dynamics of $\sigma$ over time, 
we effectively allow the adversary to \emph{choose} the value of $\sigma[t]$
in each time step, subject to some budget constraints, which can be modeled directly as a counter game; see \Cref{def:counter-game}.

\subsection{Counter Games}\label{sect:counter-games}

\begin{definition}[Counter Game]\label{def:counter-game}
The game is defined by
\begin{description}
    \item[Parameters.] Let $r,k,n_1,\ldots,n_{k-1}$ be positive integers and $c,\gamma_1,\ldots,\gamma_{k-1}$ be positive reals such that each $\gamma_i \in [c/r,1]$.

    \item[Initialization.] In the beginning $n_i'\gets n_i$,
    for each $i\in \{1,\ldots,k-1\}$.

    \item[Game Play.] In each round, based 
    on the state $(n_1',\ldots,n_{k-1}')$ the adversary plays some option $j\in [k]$.
    \begin{description}
        \item[When $j\in \{1,\ldots,k-1\}$:] With probability $\gamma_j$, set $n_j' \gets n_j'-1$.  If $n_j'<0$ the game ends.

        \item[When $j=k$:] With probability $c/r$ the game ends.
    \end{description}
    \item[Winning.] The adversary wins if the game lasts at least $r$ rounds.
\end{description}
\end{definition}

In a counter game, the adversary's optimal strategy is clear: play options $j\in[k-1]$ until the first $k-1$ counters are reduced to zero,
then play the option with the smallest probability of ending the game, which is
always $j=k$ in the cases we consider.

\begin{theorem}\label{thm:counter-game}
Fix any adversarial strategy and let $\mathcal{E}^*$ be the event 
that the adversary wins, i.e., the game lasts at least $r$ rounds.
Then
\[
\Prob{\mathcal{E}^*} \le \alpha + e^{-c(1-\frac{2}{r} \beta)},
\]
where 
$\alpha=\sum_{i} e^{-n_i/6}$ and 
$\beta=\sum_{i} \frac{n_i}{\gamma_i}$.
\end{theorem}

\begin{proof}
For each $j\in [k-1]$, let $\mathcal{E}_j$ be the event that the adversary plays option $j$ at least $2n_j/\gamma_j$ times and survives.
The expected number of decrements to counter $n_j'$ is at least $2n_j$, 
so we can upper bound $\Pr[\mathcal{E}_j] \leq e^{n_j/6}$ using a standard Chernoff bound; see~\Cref{sect:tail-bounds}.

Define $\mathcal{E}_k$ as the event in which the adversary plays option $k$ at least $r-2\beta$ times and survives. We have

\[
\Pr[\mathcal{E}_k] = (1-c/r)^{r-2\beta} \le e^{-c(1-\frac{2}{r} \beta)}.
\]

Observe that if the adversary survives $r$ rounds, at least one of the events $\mathcal{E}_1,\ldots,\mathcal{E}_k$ must occur. 
By a union bound, 
\begin{align*}
\Pr[\mathcal{E}^*] \leq \Pr[\mathcal{E}_1\cup\cdots\cup\mathcal{E}_k] 
\le \sum_{j=1}^k \Pr[E_j]
\le e^{-c(1-\frac{2}{r} \beta)} + \alpha.
\end{align*}
\end{proof}
\subsection{Upper Bounds for Expected Latency and High-probability Latency} \label{sect:upper-bounds-exp-whp}
\Cref{thm:counter-game} allows us to bound the probability that 
a party executing a \ContentionResolution{} protocol is unsuccessful 
for a period of time.  We construct a memoryless protocol $\calD$ to achieve the desired upper bound. Recall that $p(t-t_u)$ is the probability that an active party $u$ grabs the resource at local time $t-t_u$, independent of the history.

\begin{lemma}\label{lem:upper-ind-range}
Fix a memoryless protocol $\mathcal{D}$ and party $u^*\in [n]$.
Define $\mathcal{E}(u^*,t)$ to be 
the event that $u^*$ has 
yet to achieve success by global 
time $t$.  
Then
\[
\Pr\left[\mathcal{E}(u^*,a+r) \mid \mathcal{E}(u^*,a)\right] \leq
\exp\left(-c\left(1-\frac{8n(\sqrt{\log_2 n}+U)}{r\log_2\log_2 n}\right)\right)+n^{-20},
\]
where $t_{u^*}\leq a$, $p(j) \leq 1/2$ for all $j$,
$\sum_{j\in I} p(j) \leq U$ for any interval $I$ of length $r$, and $p(j) \geq \frac{4c}{r}$ for all $t\in [a,a+r)$.
\end{lemma}

\begin{proof}
Let $\sigma[t]=\sum_{u\in A[t]} p(t-t_u)$ be the
aggregate contention at global time $t$. 
Each time slot $t\in [a,a+r)$ falls into one of the following three cases.

\begin{description}
\item[Case 1:] $\sigma[t]\le 1$. By \Cref{lem:prob-of-success}, 
$u^{*}$ succeeds at time $t$ with probability at least \( p(t-t_{u^*})4^{-\sigma[t]}\ge \frac{c}{r} \).

\item[Case 2:] $1<\sigma[t]\le \frac{1}{4}\log_2\log_2 n$. Then the probability that \emph{some} party succeeds at time $t$ is at least \( \sigma[t]4^{-\sigma[t]} \ge 
\frac{\frac{1}{4}\log_2\log_2 n}{\sqrt{\log_2 n}} \).

\item[Case 3:] $\sigma[t]\ge \frac{1}{4}\log_2\log_2 n$. The probability
of success may be negligible. By definition of $U$,
$\sum_{t\in [a,a+r)} \sigma[t]\le nU$, so the \emph{number}
of times $t\in [a,a+r)$ in Case 3 is at most 
$\frac{nU}{\frac{1}{4}\log_2\log_2 n}$.
\end{description}

We map this situation onto a counter game with the following parameters.
There are $k=3$ options per play 
(corresponding to cases 2, 3, 1, respectively)
and the number of rounds is $r$;
play 1 has initial counter $n_1=n-1$ and success 
probability $\gamma_1 = \frac{1}{4}\log_2\log_2 n/\sqrt{\log_2 n}$;
play 2 has initial counter $n_2 = 4nU/\log_2\log_2 n$ 
and 
success probability $\gamma_2=1$;
play 3 ends the game with probability $c/r$.
Applying~\Cref{thm:counter-game}, we observe that
\begin{align*}
    \alpha &=e^{-n_1/6}+e^{-n_2/6} = \exp\left(-\Omega(n(1+U/\log\log n))\right) \le n^{-20},\\
    \beta &= n_1/\gamma_1+n_2/\gamma_2=\frac{n(\sqrt{\log_2 n}+U)}{\frac{1}{4}\log_2\log_2 n}.
\end{align*}
The probability that the adversary can survive $r$ rounds of the counter game is 
\[
\alpha + e^{-c(1-2\beta/r)} \leq \exp\left(-c\left(1-\frac{8n(\sqrt{\log_2 n}+U)}{r\log_2\log_2 n}\right)\right) +  n^{-20}.
\]
This is also an upper bound on the probability that $u^*$ 
fails to find a successful time slot in the interval $[a,a+r)$.
By modeling this situation as a counter game, we are effectively
letting the adversary \emph{choose} the aggregate contention 
$\sigma[t]$ as it likes, subject to a couple of budget constraints.
The aggregate contention summed over $[a,a+r)$ cannot exceed $nU$,
and the number of successes when 
$\sigma[t] \in (1,\frac{1}{4}\log_2\log_2 n]$ cannot exceed $n-1$,
for otherwise $u^*$ must have been successful.
\end{proof}

Let us now apply \Cref{lem:upper-ind-range} 
to bound the latency
of \ContentionResolution{} protocols.

\begin{theorem}[Restatement of \Cref{thm:LocalClock-exp-whp}, Part 1 and Part 2]\label{thm:expected-logn/loglogn}
There exists a memoryless acknowledgment-based \ContentionResolution{} protocol $\calD$ 
such that the expected latency 
$L_{\calD}^{\exp}(n)=O(n\log n/\log\log n)$.
\end{theorem}

\begin{proof}
The protocol $\mathcal{D}$ is defined by the probability function $p=p_{\calD}$:
\[
p(j) = \frac{1}{2^{\lceil\log_2 \lceil 1+j/10\rceil\rceil}} = \Theta(1/j) \in [0,1/2].
\]
We need to bound the expected latency for an arbitrary party $u^*\in [n]$, which, without loss of generality, is activated at time $t_{u^*}=0$. 
Let $x$ be any integer for which
$2^{x} > 100 n\log_2 n/\log_2\log_2 n$. 
Let $a$ be the first position such that $p(a)=\frac{1}{2^{x}}$. Then $p(j)=\frac{1}{2^{x}}$ for all $j\in [a,a+10\cdot 2^{x})$. 

We apply \Cref{lem:upper-ind-range} to the range 
$[a,a+r)$, where $r=10\cdot 2^{x}$.
Since $p(j)$ is monotonically nonincreasing, we can set $U=\sum_{j=1}^{r} p(j) \le \sum_{j=1}^{r} \min(\frac{10}j,\frac 12) \le 10\log_2 r$.
Since $p(j)\ge \frac{1}{2^x}$ for $j\in [a,a+r)$,
we may set $c=\frac{10}{4}$. 
Recall that $\mathcal{E}(u^*,t)$ is the event that $u^*$ has
yet to achieve success by global time $t$.
According to \Cref{lem:upper-ind-range}, we have:

\begin{align*}
\Pr\left[\mathcal{E}(u^*,a+r) \mid \mathcal{E}(u^*,a)\right] &\le \exp\left(-c\left(1-\frac{8n(\sqrt{\log_2 n}+U)}{r\log_2\log_2 n}\right)\right)+n^{-20}\\
&\le \exp\left(-\frac{10}{4}\left(1-\frac{8n(\sqrt{\log_2 n}+10\log_2 r)}{r\log_2\log_2 n}\right)\right)+n^{-20}\\
&\le \exp\left(-\frac{10}{4}\left(1-\frac{90n\log_2 r/\log_2\log_2 n}{r}\right)\right)+n^{-20}\\
&\le \exp\left(-\frac{10}{4}\cdot \frac{9}{10}\right)+n^{-20}\le 0.2.
\end{align*}

The second last inequality is due to $\frac{90n\log_2 r/\log_2\log_2 n}{r}\le \frac{1}{10}$ for sufficiently large $n$ and 
$r > 1000 n\log_2 n/\log_2\log_2 n$.

We are now prepared to analyze the expected latency of the protocol.
Let $x_0$ be the \emph{minimum} value such that 
$2^{x_0}\ge 100n\log_2 n/\log_2\log_2 n$, and 
$l_x$ as the first time slot for which 
$p(l_x-t_{u^{*}}) = \frac{1}{2^{x}}$
Then $u^*$'s expected latency is
\begin{align*}
\Expec{L_{\calD,\calA}^{(u^*)}} 
&\le \left(\sum_{1\le x'< x_0} 10\cdot 2^{x'}\right)
+\sum_{x\ge x_0} 10\cdot 2^{x}\cdot \Prob{\mathcal{E}(u^*,l_x)}\\
&= O(2^{x_0}) + \sum_{x\ge x_0} 10\cdot 2^{x}\Prob{\mathcal{E}(u^*,l_x)}\prod_{x_0< x'\le x}\Prob{\mathcal{E}(u^*,l_{x'}) \mid  \mathcal{E}(u^*,l_{x'-1})}\\
&\le O(2^{x_0})+\sum_{x\ge x_0} 10\cdot 2^{x} (0.2)^{x-x_0}\\
&\le O(2^{x_0})\le O(n\log n/\log\log n).
\end{align*}
\end{proof}

The following theorem was first claimed in the extended abstract
of De Marco and Stachowiak~\cite{MarcoS17}, and later proved
by De Marco, Kowalski, and Stachowiak~\cite{demarco2022timeenergyefficientcontention}.
Here we provide an alternate proof that depends on \Cref{lem:upper-ind-range}'s reduction to a counter game.

\begin{theorem}[Cf.~\cite{demarco2022timeenergyefficientcontention,MarcoS17}]\label{thm:whp-latency-upper-bound}
There exists a memoryless \ContentionResolution{} protocol 
that guarantees latency $O(n\log^2 n/\log\log n)$ with high probability.
\end{theorem}

\begin{proof}
The protocol is defined by the memoryless distribution $\calD$, where 
\[
p(j)=\frac{\lceil\log_2 \lceil1+j/10\rceil\rceil}{2^{\lceil \log_2 \lceil1+j/10\rceil\rceil}} = \Theta(\log j/j) \in [0,1/2],
\] 
for all $j\ge 1$. Let $u^{*}$ be an arbitrary party.

Let $x$ be any integer such that $2^{x} > 100 n\log_2^2 n/\log_2\log_2 n$. Let $a$ be the first position with $p(a-t_{u^*})=\frac{x}{2^x}$, and $r=10\cdot 2^{x}$.
We will apply \Cref{lem:upper-ind-range} to the interval 
$[a,a+r)$. Since $p(j)$ is monotonically non-increasing for integer $j\ge 1$, we can set 
\[
U=\sum_{j=1}^{r} p(j)\le \sum_{x=1}^{\lceil \log_2 \lceil 1+r/10\rceil\rceil} \frac{x}{2^x}\cdot 10\cdot 2^x\le 10\log_2^2 r.
\]
To conform to the requirements of \Cref{lem:upper-ind-range}, 
we set $c=\frac{10x}{4}$ so that 
$p(j)=\frac{x}{2^x}=\frac{4c}{r}$.
Now applying \Cref{lem:upper-ind-range}, we have

\begin{align*} 
\Prob{\mathcal{E}(u^*,a+r)}
&=\Prob{\mathcal{E}(u^*,a)}\cdot \Prob{\mathcal{E}(u^*,a+r) \mid \mathcal{E}(u^*,a)}\\
&\le \exp\left(-c\left(1-\frac{8n(\sqrt{\log_2 n}+U)}{r\log_2\log_2 n}\right)\right)+n^{-20}\\
&\le \exp\left(-\frac{10x}{4}\left(1-\frac{90n\log_2^2 r/\log_2\log_2 n}{r}\right)\right)+n^{-20}\\
&\le \exp\left(\mbox{$-\frac{10 x}{4}\frac{9}{10}$}\right)+n^{-20}\\
&=\exp\left(-\mbox{$\frac{9}{4}$} \log_2 \lceil 1+r/10\rceil\right)+n^{-20}\\
&\le \exp(-2\ln n)=n^{-2}.
\end{align*}
\end{proof}

\section{Conclusion and Open Problems}\label{sect:conclusion}

The first contribution of this paper is to explore the power of the \GlobalClock{} model for solving randomized \ContentionResolution,
in which parties perceive both the global clock and their local time since wake-up.  We proved that \GlobalClock{} is strictly stronger than \LocalClock{} by exhibiting an \emph{Elias code}-based protocol that achieves
near-linear latency
\[
O(n\zeta(4\log\log n))=O(n\log\log n\log^{(3)}n\cdots 2^{\log^* n}) = n(\log\log n)^{1+o(1)}.
\]
The most pressing question is now to prove or disprove \Cref{conj:global-clock}(1), i.e., 
whether there exists a \GlobalClock{} protocol with $O(n)$ latency, 
as well as its analogue \Cref{conj:global-clock}(2) 
for the unbounded setting.

\medskip 

In the \LocalClock{} model we established sharp bounds on the latency of memoryless protocols under both the In-Expectation and With-High-Probability objectives, and also proved that optimality under both objectives 
cannot be achieved simultaneously.
We believe the memoryless assumption is not critical, and that all of our lower bounds can be extended
to arbitrary acknowledgment-based protocols.  
It is well known that proving lower bounds for acknowledgment-based protocols is more difficult~\cite{GoldbergL25,MarcoS17,demarco2022timeenergyefficientcontention}, 
despite there being no obvious advantage to 
techniques that violate memorylessness.


\newcommand{\etalchar}[1]{$^{#1}$}

\appendix

\section{Asymptotic Notation: The Definition of \texorpdfstring{$\Omega$}{Ω}}\label{sect:asymptotic-notation-controvery}

Is saying $g$ cannot be $o(f)$ the same as saying $g=\Omega(f)$?  This goes back to a minor controversy of the 1970s on the proper definition of asymptotic notation for non-negative $\mathbb{Z}^+\to\mathbb{R}^+$ functions, particularly the definition of $\Omega$.  
    Everyone agreed that
    $g=O(f)$ means $g(n)\leq cf(n)$ for some $c>0$ and all sufficiently large $n$.  
    Aho, Hopcroft, and Ullman's 1974 text~\cite{AHU74} defined $\Omega$ as:
    \[
    g = \Omega_{\mathrm{AHU}}(f) 
    \mbox{ if $g(n) \geq cf(n)$ for some $c>0$ and \emph{infinitely many $n$}.}
    \]
    In 1976 Knuth~\cite{Knuth76} published a history of asymptotic notation and dated 
    the Aho-Hopcroft-Ullman definition of 
    $\Omega$ to a 1914 paper of Hardy and Littlewood~\cite{HardyL14}.
    Without advancing a specific argument, Knuth~\cite{Knuth76} 
    stated that a definition of $\Omega$ symmetric to $O$
    would prove to be most useful for computer scientists, namely 
    \[
    g=\Omega_{\mathrm{K}}(f)
    \mbox{ if $g(n)\geq cf(n)$ for some $c>0$ and all sufficiently large $n$.}
    \]
    Although Knuth's definition $\Omega_{\mathrm{K}}$ is most common in today's textbooks~\cite{CLRS22,KT05,DPV06}, most computer scientists operationally 
    use $\Omega$ to mean either $\Omega_{\mathrm{AHU}}$ or $\Omega_{\mathrm{K}}$, depending on context.
    
    The problem, as Aho, Hopcroft, and Ullman discussed~\cite{AHU74}, 
    is that some problems are not expected to be uniformly and simultaneously
    hard for all values of $n$.
    This is certainly true in \ContentionResolution, 
    where approximating ``$n$'' is the crux of the 
    problem.  For example, De Marco, Kowalski, and Stachowiak~\cite{demarco2022timeenergyefficientcontention} 
    assert a lower bound of $\Omega(n\log n/(\log\log n)^2)$, 
    \emph{by which they mean} $\Omega_{\mathrm{AHU}}(n\log n/(\log\log n)^2)$.
    It is perfectly consistent with~\cite{demarco2022timeenergyefficientcontention}
    and our lower bounds that there is a \ContentionResolution{} scheme with latency, 
    say, $O(n\log\log\log n)$ whenever $n$ is of the form $2^{2^{2^{2^k}}}$, 
    making it essentially impossible to say something meaningful using 
    Knuth's $\Omega_{\mathrm{K}}$ notation.
    
    In our view, when describing the asymptotic complexity of problems, $\Omega_{\mathrm{AHU}}$ 
    is more attractive 
    than $\Omega_{\mathrm{K}}$. 
    It is nearly always more useful for ``$\Omega$'' to be the logical negation of ``$o$'' than the logical converse of ``$O$.''

\section{Tail Bounds}\label{sect:tail-bounds}

\begin{lemma}[Classical Azuma-Hoeffding Bound, see e.g. \cite{DubhashiPanconesi09}] \label{lem:cla-azuma}
    Let $X_1, \dots, X_n$ be independent random variables such that $X_i \in [a_i, b_i]$ almost surely. Let $\mu=\sum_{i=1}^{n} \Expec{X_i}$, then for any $\epsilon > 0$,
    \[
    \Pr\left[\left(\sum_{i=1}^n X_i\right) \ge (1+\epsilon) \mu\right],
    \Pr\left[\left(\sum_{i=1}^n X_i\right) \le (1-\epsilon) \mu\right] 
    \le \exp\left(-\frac{2\epsilon^2\mu^2}{\sum_{i=1}^n (b_i - a_i)^2}\right).
    \]
\end{lemma}
\begin{lemma}[Chernoff Bound]\label{lem:chernoff-bound}
    Let $X_1,\ldots,X_n$ be independent random variables with range $[0,v]$. Let $\mu=\sum_{i=1}^{n} \Expec{X_i}$, then for all $\epsilon>0$, 
\begin{enumerate}
\item $
    \Prob{\left(\sum_{i=1}^{n} X_i\right)\geq (1+\epsilon)\mu}\leq \left(\frac{e^{\epsilon}}{(1+\epsilon)^{1+\epsilon}}\right)^{\mu/v}
    $
\item $
     \Prob{\left(\sum_{i=1}^{n} X_i\right)\leq (1-\epsilon)\mu}\leq \left(\frac{e^{-\epsilon}}{(1-\epsilon)^{1-\epsilon}}\right)^{\mu/v}
    $
\end{enumerate}
\end{lemma}

\begin{corollary}\label{coro:chernoff-bound-2}
Let $X_1,\ldots,X_n$ be independent random variables supported on $[0,v]$, 
and let $\mu = \sum_{i=1}^{n} \Expec{X_i}$. Then the following inequalities hold:
\begin{enumerate}
    \item $\Prob{\sum_{i=1}^{n} X_i \ge 2\mu} \le \exp\!\left(-\frac{\mu}{3v}\right)$,
    \item $\Prob{\sum_{i=1}^{n} X_i \le \mu/2} \le \exp\!\left(-\frac{\mu}{7v}\right)$.
\end{enumerate}
\end{corollary}

\begin{proof}
\mbox{}\\
\noindent\emph{Proof of Part 1.}
Applying~\Cref{lem:chernoff-bound}~(i) with $\epsilon = 1$, we obtain
\[
\left(\frac{e^{\epsilon}}{(1+\epsilon)^{1+\epsilon}}\right)^{\mu/v}
= \left(\frac{e}{2^2}\right)^{\mu/v}
\le e^{-\mu/(3v)}.
\]

\noindent\emph{Proof of Part 2.}
Applying~\Cref{lem:chernoff-bound}~(ii) with $\epsilon = 1/2$, we obtain
\[
\left(\frac{e^{-\epsilon}}{(1-\epsilon)^{1-\epsilon}}\right)^{\mu/v}
= \left(\frac{e^{-1/2}}{(1/2)^{1/2}}\right)^{\mu/v}
\le e^{-\mu/(7v)}.
\]
\end{proof}

\section{Omitted Proofs}

\subsection{Proof of \Cref{lem:prob-of-success}}\label{sect:prob-of-success}

Let us recall the two claims of \Cref{lem:prob-of-success}.
For a global time $t$,  
$p_u = \Prob{X_{t-t_u}^{(u)}=1}$ 
is the probability that $u$ 
grabs the shared resource and 
$\sigma = \sigma[t] = \sum_{u\in A[t]} p_u$.

\begin{enumerate}
    \item If $p_u \in [0,1/2]$ for all $u\in A[t]$, then $\Prob{\left(\sum_{u\in A[t]} X_{t-t_u}^{(u)}\right)=1}\geq \sigma 4^{-\sigma}$, 
    and $u$'s probability of success is at 
    least $p_u4^{-\sigma}$.
    \item $\Prob{\left(\sum_{u\in A[t]} X_{t-t_u}^{(u)}\right)=1}\le \sigma e^{-\sigma+1}$
    \ and \ $\Prob{\left(\sum_{u\in A[t]} X_{t-t_u}^{(u)}\right)\le 1}\le e^{-\sigma}+\sigma e^{-\sigma+1}$.
\end{enumerate}

\begin{proof}[Proof of \Cref{lem:prob-of-success}]
\mbox{}\\
\noindent
\textbf{Part 1.}
We use the approximation $(1-x)\geq 4^{-x}$ for $x\in [0,1/2]$. 
The probability of success is:
\begin{align*}
\Prob{\left(\sum_{u\in A[t]} X_{t-t_u}^{(u)}\right)=1} 
    &= \sum_{u\in A[t]} p_u \prod_{v\in A[t]-\{u\}} (1-p_v)
    \ge \sigma 4^{-\sigma}.
\end{align*}
Moreover, $u$'s probability of success is $p_u\prod_{v\in A[t]-\{u\}}(1-p_v) \geq p_u4^{-\sigma}$.

\noindent
\textbf{Part 2.}
We know $1-x\le e^{-x}$ for any $x$, so 
\begin{align*}
\Prob{\left(\sum_{u\in A[t]} X_{t-t_u}^{(u)}\right)=1} &= \sum_{u\in A[t]} p_u\prod_{v\in A[t]\setminus \lbrace u\rbrace}(1-p_v) \leq \sigma e^{-\sigma+1},\\
\Prob{\left(\sum_{u\in A[t]} X_{t-t_u}^{(u)}\right)=0} &=\sum_{u\in A[t]} (1-p_u)\le e^{-\sigma}.
\end{align*}
\end{proof}
Note that when $\sigma\geq 2$, it is impossible to get a non-zero lower bound on the probability of success without imposing some non-trivial upper bound on the probabilities, such as $p_i \in [0,1/2]$.

\subsection{Proof of \Cref{lem:n0}(v)}\label{sect:smoothness-proofs}

\Cref{lem: smoothness of U} and \Cref{coro: smoothness of U} prove part (v) of \Cref{lem:n0}.

\begin{lemma}\label{lem: smoothness of U}
    Suppose $f(n)$ is nonnegative, monotone nondecreasing and $f(n+1)-f(n)\le f(n+1)/2$ when $n\ge n_0$. Fix some $0\le c<1$ and integer $n\ge n_0$. Let $n'$ be the largest integer such that $f(n')\le cf(n)$. If $f(n_0)\le c\cdot f(n)$, then $f(n')\ge cf(n)/2$.
\end{lemma}
\begin{proof}
    Since $f(n_0)\le cf(n)$, $n'\ge n_0$. By the definition of $n'$, we have $f(n'+1)>cf(n)$. So
    $$
    f(n')=f(n'+1)-(f(n'+1)-f(n'))\ge f(n'+1)/2> cf(n)/2.
    $$
\end{proof}

\begin{corollary}\label{coro: smoothness of U}
Let $q(n)$ be either $q(n)=1/2$ or $q(n)=n^{-2}$, 
and recall that
$U(n)=\lfloor \frac{n\ln(1/q(n))\ln n}{\ln\ln n}\rfloor$. 
For sufficiently large $n$ and a real number $c\in [20/n,1)$, let $n'$ be the largest integer such that $U(n')\le cU(n)$. Then
$$
U(n')\ge cU(n)/2.
$$
\end{corollary}

\begin{proof}
    Clearly $U(n+1)-U(n)=O(\frac{\ln(1/q(n))\ln n}{\ln\ln n})\le U(n+1)/2$ and $U(n)$ is monotone increasing when $n$ is large, say $n\ge n_0$. Since $c\ge 20/n$, $cU(n)=\Omega(\frac{\ln(1/q(n))\ln n}{\ln\ln n})\ge U(n_0)$ for sufficiently large $n$. The conclusion follows from \Cref{lem: smoothness of U}.
\end{proof}
\end{document}